\documentclass[10pt,a4paper]{article} %***
\usepackage[sectionbib]{natbib}
\usepackage{array,epsfig,fancyheadings,rotating}
\usepackage[]{hyperref}  %<----modified by Ivan
\usepackage{sectsty, secdot}
\sectionfont{\fontsize{12}{14pt plus.8pt minus .6pt}\selectfont}

\subsectionfont{\fontsize{12}{14pt plus.8pt minus .6pt}\selectfont}
\usepackage{amsmath}
\usepackage{amssymb}
\usepackage{amsfonts}
\usepackage{multirow}
\usepackage{amsthm}

\usepackage{kashsty}
\usepackage{kashThm}

\usepackage[normalem]{ulem}
\def\PICDIR{./}

\def\strue{\Sigma}
\def\espar{\hat{\Sigma}^\mathrm{sp}}
\def\esparij{\hat{\sigma}^\mathrm{sp}_{i,j}}
\def\ediag{\hat{\Sigma}^{\mathrm{diag}}}
\def\eemp{\hat{\Sigma}^\mathrm{emp}}
\def\estar{\hat{\Sigma}^{\star}}
\def\etild{\tilde{\Sigma}}
\def\ehalf{\hat{\Sigma}_{0.5}}
\def\targ{{\Sigma}^\mathrm{tar}}

\def\phn{\phantom{\LARGE X}}
\def\ezero{\hat{\Sigma}_{00}}
\def\efifty{\hat{\Sigma}_{50}}

\begin{document}

%%%%%%%%%%%%%%%%%%%%%%%%%%%%%%%%%%%%%%%%%%%%%%%%%%%%%%%%%%%%%%%%%%%%%%%%%%%%%%%%%%%%%%%%%%%%%%%%%%%%%%%%%%%%%%%%%%%%%%%%%%%%
%%%%%%%%%%%%%%%%%%%%%%%%%%%%%%%%%%%%%%%%%%%%%%%%%%%%%%%%%%%%%%%%%%%%%%%%%%%%%%%%%%%%%%%%%%%%%%%%%%%%%%%%%%%%%%%%%%%%%%%%%%%%

%%%%%%%%%%%%%%%%%%%%%%%%%%%%%%%%%%%%%%%%%%%%%%%%%%%%%%%%%%%%%%%%%%%%%%%%%%%%%%%%%%%%%%%%%%%%%%%%%%%%%%%%%%%%%%%%%%%%%%%%%%%%

%%%%%%%%%%%%%%%%%%%%%%%%%%%%%%%%%%%%%%%%%%%%%%%%%%%%%%%%%%%%%%%%%%%%%%%%%%%%%%%%%%%%%%%%%%%%%%%%%%%%%%%%%%%%%%%%%%%%%%%%%%%%

% "Title of the paper"
%\title{Nonasymptotic estimation and support recovery for 
%high dimensional sparse covariance matrices}
%\runtitle{Nonasymptotic support recovery and estimation}

% indicate corresponding author with \corref{}
% \author{\fnms{John} \snm{Smith}\corref{}\ead[label=e1]{smith@foo.com}\thanksref{t1}}
% \thankstext{t1}{Thanks to somebody} 
% \address{line 1\\ line 2\\ printead{e1}}
% \affiliation{Some University}

%\author{\fnms{Adam B} \snm{Kashlak}\corref{}\ead[label=e1]{kashlak@ualberta.ca}\ead[label=e2]{https://sites.ualberta.ca/\textasciitilde kashlak/}}
%\and
%\author{\fnms{Linglong} \snm{Kong}\ead[label=e3]{lkong@ualberta.ca}}
%\affiliation{University of Alberta}
%\address{Department of Mathematical and Statistical Sciences, CAB 632\\
%    University of Alberta\\
%    Edmonton, AB, Canada T6G 2G1  \\
%    \printead{e1}\\
%    \printead{e2}\\
%}
%\author{\fnms{???} \snm{???}\ead[label=e2]{???}}
%\address{\printead{e2}}
%\affiliation{University of Alberta}

%\runauthor{AB Kashlak and L Kong}

\title{Nonasymptotic Estimation and Support Recovery\\for 
High Dimensional Sparse 
Covariance Matrices
}

\author{Adam B Kashlak kashlak@ualberta.ca\\
Linglong Kong lkong@ualberta.ca\\
Department of Mathematical and Statistical Sciences\\
University of Alberta\\
Edmonton, AB, Canada T6G 2G1}

\maketitle

\begin{abstract}
  We propose a general framework for nonasymptotic covariance
  matrix estimation making use of concentration inequality-based confidence
  sets.  We specify this framework for the 
  estimation of large sparse covariance matrices through incorporation
  of past thresholding estimators with key emphasis on support recovery.
  This technique goes beyond 
  past results for thresholding estimators 
  as we have distribution free control over the false positive rate 
  being the number of entries incorrectly included in the estimator's
  support.
  In the context of support recovery, we are able to specify a 
  false positive rate and optimize to maximize the true 
  recoveries.
  This methodology guarantees exact support recovery in the 
  case of strongly log concave data and maintains good performance
  in more general distributional settings.
  The usage of nonasymptotic dimension-free confidence sets 
  yields good theoretical performance.
  Through extensive simulations, it is demonstrated 
  to have superior performance when compared with other such methods.

  \vspace{9pt}
  \par
\end{abstract}

%\begin{keywords}
%  concentration inequality, confidence region,
%  random matrix, schatten norm, log concave measure
%\end{keywords}

\section{Introduction}

Covariance matrices and accurate estimators of such objects
are of critical importance in statistics.  
Various standard techniques including principal components 
analysis and linear and quadratic discriminant analysis 
rely on an accurate estimate of the covariance structure of the data.
Applications can range from genetics and medical imaging data
to climate and other types of data.
Furthermore, in the
era of high dimensional data, classical asymptotic estimators 
perform poorly in applications \citep{STEIN1975,JOHNSTONE2001}. 
Hence, we propose a general methodology for nonasymptotic 
covariance matrix estimation making use of confidence balls 
constructed from concentration inequalities.  While this
is a general framework with many potential applications,
we specifically consider the use of thresholding estimators
for sparse covariance matrices with a view towards support
recovery---that is, determining which variable pairs are
correlated.

Many estimators
for the covariance matrix have been proposed working under the
assumption of sparsity \citep{POURAHMADI2011}, which is,
in a qualitative sense, the case when 
most of the off-diagonal entries are zero or negligible.  
Beyond mere theoretical interest, the assumption of sparsity is 
widely applicable to real data analysis as the practitioner may
believe that many of the variable pairings will be uncorrelated.
Thus, it is desirable to tailor covariance estimation procedures
given this assumption of sparsity.

Sparsity in the simplest sense implies some bound on the 
number of non-zero entries in the columns of a covariance 
matrix.  Thus, given a $\Sigma\in\real^{d\times d}$ with
entries $\sigma_{i,j}$ for $i,j=1,\ldots,d$, 
there exists some constant $k>0$ such that
$\max_{j=1,\ldots,d}\sum_{i=1}^d \Indc{\sigma_{i,j}\ne 0} \le k$.
This can be generalized to ``approximate sparsity'' as 
in \cite{ROTHMAN2009} by 
$\max_{j=1,\ldots,d}\sum_{i=1}^d \abs{\sigma_{i,j}}^q \le k$ for
some $q\in[0,1)$.  Furthermore, \cite{CAILUI2011} define a 
broader approximately sparse class by bounding weighted column
sums of $\Sigma$.  In \cite{KAROUI2008}, a similar notion 
referred to as ``$\beta$-sparsity'' is defined.
Such classes of sparse covariance matrices allow for 
good theoretical performance of estimators.
 
One class of estimators are shrinkage estimators 
that follow a James-Stein approach by shrinking 
estimated eigenvalues, eigenvectors, or the matrix itself
towards some desired target
\citep{HAFF1980,DEY1985,DANIELS1999,DANIELS2001,LEDOIT2004,HOFF2009,JOHNSTONE2012}.
Another class of sparse estimators are those that
regularize the estimate with lasso-style penalties
\citep{ROTHMAN2012,BIEN2011}.  Yet another class consists of 
thresholding estimators, which declare the covariance between
two variables to be zero, if the estimated value is smaller 
than some threshold 
\citep{BICKELLEVINA2008,BICKELLEVINA2008A,ROTHMAN2009,CAILUI2011}.
Beyond these, there are other methods such as banding and
tapering, which apply only when the variables are ordered
or a notation of proximity exists---for example, spatial, time series, 
or longitudinal data.  As we will 
not assume such an ordering and strive to construct a methodology that
is permutation invariant with respect to the variables, 
these approaches will not be considered.  Lastly, there has also 
been substantial work 
into the estimation of the precision or inverse covariance matrix.
While it is easily possible that our approach could be adapted to this
setting, it will not be considered in this article and will, hence,
be reserved for future research.

In this article, we propose of novel approach to the 
estimation of sparse covariance matrices making use of 
concentration inequality based confidence sets such as 
those constructed in~\cite{KASHLAK_COV2018} for the functional
data setting.  
%This approach takes inspiration from the shrinkage, thresholding, and
%penalization methods of sparse covariance estimation.
In short, consider a sample of real vector valued
data $X_1,\ldots,X_n\in\real^d$ with mean zero and unknown covariance
matrix $\strue$.  Concentration inequalities are 
used to construct a non-asymptotic confidence set for $\strue$ about
the empirical estimate of the unknown covariance matrix, 
$\eemp=n^{-1}\sum_{i=1}^n (X_i-\bar{X})\TT{(X_i-\bar{X})}$
where $\bar{X} = n^{-1}\sum_{i=1}^nX_i$ is the sample mean.
While, it has been noted---for example, see~\cite{CAILUI2011}---that 
$\eemp$ may be a poor estimator when the dimension $d$ is
large and $\strue$ is sparse, the confidence set is still valid 
given a desired coverage of $(1-\alpha)$.  
To construct a better 
estimator, we propose to search this confidence set for an 
estimator $\espar$
which optimizes some sparsity criterion to be concretely defined later.
This estimation method adapts to the uncertainty of $\eemp$ in the 
high dimensional setting, $d\gg n$, by widening the confidence
set and thus allowing our sparse estimator to lie far away
from the empirical estimate.
Furthermore, given some distributional assumptions, 
the concentration inequalities 
provide us with non-asymptotic dimension-free confidence sets 
allowing for very desirable convergence results.

Many established methods for sparse estimation make use of 
a regularization or penalization term incorporated to enforce sparsity
\citep{ROTHMAN2012,BIEN2011}.  
In some sense, our proposed method
can be considered to be in this class of estimators.  However, we do not 
enforce sparsity via some lasso-style penalization term, but enforce it
by
\renewcommand{\theenumi}{\roman{enumi}}
\begin{enumerate}
 \item choosing a desired false positive rate, $0<\rho\ll1$, 
       for the support recovery,
 \item using that rate to construct a $(1-\alpha)$ confidence ball about the 
       empirical estimator, and
 \item searching that ball for a sparse estimator.
\end{enumerate}
The larger our
$(1-\alpha)$-confidence set is, the sparser our estimator 
is allowed to be.  Thus, the radius of our confidence balls
acts like a regularization parameter allowing for greater 
sparsity as it increases.
A major contribution of this work is developing a method
with the ability to avoid
costly cross-validation of the tuning parameter and maintain
strong finite sample performance.
The specific focus as discussed below and in the supplementary material
%Appendix~\ref{app:empDiag}
is accurate support recovery, which is the identification of the non-zero
entries in the covariance matrix.  Our methodology allows for 
fixing a false positive rate---percentage of zero entries incorrectly 
said to be non-zero---and optimizing over the true positive 
rate---percentage of correctly identified non-zero entries.
Furthermore, our estimation technique implements a binary 
search procedure resulting in a highly efficient algorithm 
especially when compared to the more laborious optimization
required by lasso penalization.

\iffalse
In some sense, our proposed method
can be considered to be in this class of estimators.  However, we do not 
enforce sparsity via some lasso-style penalization term, but enforce it
through the choice of $\alpha$.  The larger our
$(1-\alpha)$-confidence set is, the sparser our estimator 
is allowed to be.  Thus, as with other regularized estimators, the
practitioner will have to make a choice of $\alpha$ 
to achieve the desired level of sparsity.
A major contribution of this work is developing a method
with the ability to avoid
costly cross-validation of the tuning parameter and maintain
strong finite sample performance.
The specific focus as discussed below and in the supplementary material
%Appendix~\ref{app:empDiag}
is accurate support recovery, which is the identification of the non-zero
entries in the covariance matrix.  Our methodology allows for 
fixing a false positive rate---percentage of zero entries incorrectly 
said to be non-zero---and optimizing over the true positive 
rate---percentage of correctly identified non-zero entries.
Furthermore, our estimation technique implements a binary 
search procedure resulting in a highly efficient algorithm 
especially when compared to the more laborious optimization
required by lasso penalization.
\fi

In Section~\ref{sec:sparseEstProc}, the general estimation 
procedure is outlined, and it is specified for tuning threshold
estimators with concentration methods.  
Section~\ref{sec:falsePos} discusses our approach to fixing 
a certain false positive rate when attempting to recover the
support of the covariance matrix.
In Section~\ref{sec:sparseEstCases}, three
different types of concentration are considered for specifically
log concave measures, sub-exponential distributions, and bounded 
random variables.  
Lastly, Section~\ref{sec:numeric}
details comprehensive simulations comparing our concentration approach
to sparse estimation to standard techniques such as 
thresholding and penalization.  Beyond simulation experiments,
a real data set of gene expressions for small round blue cell
tumours from the study of \cite{KHAN2001} is considered.

\subsection{Notation and Definitions}

We will make use of both a $(1-\alpha)$-confidence set and
a false positive rate $\rho$.  For the former, we have the 
usual definition that some data dependent set 
$\mathcal{C}_{1-\alpha}$ is a
$(1-\alpha)$-confidence set for $\Sigma$ if
$
  \prob{ \Sigma \notin \mathcal{C}_{1-\alpha} } \le \alpha.
$
For an estimator of $\Sigma$ in $\real^{d\times d}$, we have to
decide which of the $d(d-1)/2$ off-diagonal entries are non-zero.
The false positive rate $\rho$ is the probability that we 
incorrectly decide that a given entry is non-zero.

When defining a Banach space of matrices,
there are many matrix norms that can be considered.
In the article, the main norms of interest are the 
$p$-Schatten norms, which
will be denoted $\norm{\cdot}_p$ and are defined as follows.
\begin{definition}[$p$-Schatten Norm]
  For an arbitrary matrix $\Sigma\in\real^{k\times l}$
  and $p\in[1,\infty)$, 
  the $p$-Schatten norm is
  $
    \norm*{\Sigma}_p = \tr{(\TT{\Sigma}\Sigma)^{p/2}}^{1/p}
    = \norm{{\boldsymbol \nu}}_{\ell^p}
    = \left(\sum_{i=1}^{\min\{k,l\}} \nu_i^p\right)^{1/p}
  $
  where ${\boldsymbol \nu}=(\nu_1,\ldots,\nu_{\min\{k,l\}})$
  is the vector of singular values of $\Sigma$ and where
  $\norm{\cdot}_{\ell^p}$ is the standard $\ell^p$ norm 
  in $\real^d$.
  In the covariance matrix case where $\Sigma\in\real^{d\times d}$
  is symmetric and positive semi-definite,
  $
    \norm*{\Sigma}_p = \tr{\Sigma^p}^{1/p}
    = \norm{{\boldsymbol \lmb}}_{\ell^p}
  $
  where ${\boldsymbol \lmb}$ is the vector of eigenvalues of $\Sigma$.
  The $1$-Schatten norm is referred to as the trace norm and 
  the $2$-Schatten norm as the Hilbert-Schmidt
  or Frobenius norm.

  For $p=\infty$, we have the usual operator norm for 
  $\Sigma:\real^{l}\rightarrow\real^k$ with respect to the 
  $\ell^2$ norm,
  $
    \norm*{\Sigma}_\infty 
    = \sup_{\norm{u}_{\ell^2}=1} \norm{\Sigma u}_{\ell^2}
    = \norm{{\boldsymbol \nu}}_{\ell^\infty}
    = \max_{i=1,\ldots,\min\{k,l\}} \abs{\nu_i},
  $
  which is similarly the maximal eigenvalue when $\Sigma$ is
  symmetric positive semi-definite.
\end{definition}

The definition of the $p$-Schatten norm 
involves taking the square root of a
symmetric matrix.  In general, a matrix square root
is only unique up to unitary transformations.  However,
for symmetric positive semi-definite matrices, we will only
require the unique symmetric positive semi-definite square
root defined as follows.
\begin{definition}[Matrix Square Root]
  Let $A\in\real^{d\times d}$ be a symmetric positive semi-definite
  matrix
  with eigen-decomposition $A = UD\TT{U}$ 
  where $U$ the orthonormal matrix of eigenvectors and $D$ the diagonal
  matrix of eigenvalues, $(\lmb_1,\ldots,\lmb_d)$.  
  Then, $A^{1/2} = UD^{1/2}\TT{U}$
  where $D^{1/2}$ is the diagonal matrix with entries
  $(\lmb_1^{1/2},\ldots,\lmb_d^{1/2})$.
\end{definition}

Another family of norms that will be used is the collection of
\textit{entrywise} matrix norms denoted, which are written in
terms of $\ell^p$ norms of the entries. 
\begin{definition}[$(p,q)$-Entrywise norm]
  \label{defn:entrywiseNorm}
  For an arbitrary matrix $\Sigma\in\real^{k\times l}$ with
  entries $\sigma_{i,j}$
  and $p,q\in[1,\infty]$, the $(p,q)$-entrywise norm is
  $
    \norm{\Sigma}_{p,q} = \left[
      \sum_{i=1}^k(\sum_{j=1}^l\sigma_{i,j}^{q})^{p/q}
    \right]^{1/p}
  $
  with the usual modification in the case that $p=\infty$ and/or
  $q=\infty$.
  When $p=q$,
  these are the $\ell^p$ norms of a given matrix treated as a
  vector in $\real^{k,l}$.  
  Note that the 2-Schatten norm coincides with the $(2,2)$-entrywise
  norm.
\end{definition}

\subsection{Main contributions and connections to past work}

The main contribution of this work is the construction of a 
general framework for tuning threshold estimators for support
recovery and estimation of sparse 
covariance matrices.  
%This framework has the potential to extend to other
%distributional assumptions and to other estimation techniques.
It offers finite sample guarantees and a much faster compute
time than computationally expensive optimization
and cross validation methods.

Past work on thresholding estimators for 
sparse covariance estimation began with solely
considering Gaussian data and then extending to sub-Gaussian tails
\citep{BICKELLEVINA2008,BICKELLEVINA2008A,ROTHMAN2009}.
The more recent work of \cite{CAILUI2011} also provides
theoretical results for sub-Gaussian data as well as 
certain polynomial-type tails.  However, only Gaussian
data is considered in their numerical simulations.
In this article, we consider strongly log-concave, 
heavier tailed sub-exponential, and bounded data.
While bounded data is, in fact, sub-Gaussian, the 
concentration behaviour of such data may be dependent on
the dimension of the space compared to the much better 
behaved strongly log concave measures that also exhibit
sub-Gaussian concentration.

The principal focus of this work is to use 
non-asymptotic concentration inequalities to guarantee
finite sample performance.  Past articles are focused 
on proximity of their estimator to truth in operator
norm as the main metric of success due to convergence in 
operator norm implying convergence of the eigenvalues and 
eigenvectors.  While asymptotically,
such methods have elegant theoretical convergence properties, for 
finite samples one can achieve better performance in operator
norm distance by simply choosing the empirical diagonal matrix 
as an estimator---that is, the empirical estimator with off-diagonal
entries set to zero.  In the supplementary material,
%Appendix~\ref{app:empDiag}, 
we rerun 
some of the numerical simulations from \cite{ROTHMAN2009}
and demonstrate that for Gaussian data the empirical
diagonal matrix achieves better performance than all 
of the universal threshold estimators for data in $\real^{500}$
for a sample of size $n=100$.
For sub-exponential data---albeit outside of the scope
of their---the empirical diagonal matrix 
dominated all threshold estimators in operator norm distance
even when $d<n$.
We thus strongly argue that the main metric of success 
for such sparse estimators is support recovery
of the true covariance matrix.

The main theoretical results of this work are
Theorem~\ref{thm:falsepos}, which establishes
how to fix a false positive rate for threshold estimators
devoid of any distributional assumptions,
and
Theorem~\ref{thm:logConSpar}, which demonstrates 
support recovery---both zero and non-zero entries---in 
the specific case that the data has a
strongly log concave measure.  
In the case that the
data is instead sub-exponential or bounded, we do not achieve 
a similar limit theorem, but are still able to achieve
good performance in numerical simulations.
Of independent interest is Lemma~\ref{lem:covSymmet},
which establishes a symmetrization result for
sparse random matrices making use of the techniques
in \cite{LATALA2005}.

\iffalse
For support recovery, let $\strue$ be the true covariance
matrix and $\espar$ be the sparse estimator.  We define
a true positive to be a pair $(i,j)$ with $i,j=1,\ldots,d$
such that the $i,j$th entries $\sigma_{i,j}\ne0$ and 
$\esparij\ne0$.
We similarly define a false positive to be a pair $(i,j)$
such that $\sigma_{i,j}=0$ and $\esparij\ne0$.
As often in statistical practice, we aim to maximize the
true positive rate while controlling the false positive rate.
Unlike other methods that only provide a single estimator
with no ability to control the false positive rate,
our framework can be tuned to a desired false positive
rate.
\fi

\section{Sparse Estimation Procedure}
\label{sec:sparseEstProc}

Let $X_1,\ldots,X_n\in\real^d$ be a sample of $n$ \iid
mean zero random vectors with 
unknown $d\times d$ covariance matrix $\strue$.
Define the empirical estimate of $\strue$ to be
$\eemp = n^{-1}\sum_{i=1}^n(X_i-\bar{X})\TT{(X_i-\bar{X})}$
where $\bar{X} = n^{-1}\sum_{i=1}^nX_i$.
The goal of the following procedure is to construct
a sparse estimator, $\espar$, for $\strue$ by first constructing
a non-asymptotic confidence set for $\strue$ centred on $\eemp$ 
and then searching this set for the sparsest member.
A search method using threshold estimators is outlined in 
Section~\ref{sec:zeroMeth}.  
%A different approach using geodesics in the 
%Procrustes metric is considered in Appendix~\ref{sec:procMeth}.
%This latter method did not perform as well in the numerical
%simulation and thus is relgated to the appendix.

The methodology is as follows:
\begin{enumerate}
  \item Choose a suitable false positive rate $\rho\in(0,1)$,
    which will typically be close to zero.
  \item Use Theorem~\ref{thm:falsepos} to determine the 
    radius of a ball centred at $\eemp$ such that the 
    sparsest matrices in that ball have false positive rate $\rho$.
  \item Use the binary search algorithm in Section~\ref{sec:zeroMeth}
    to identify the sparsest element in the above ball denoted
    $\espar$.
  \item Considering this ball as a $(1-\alpha)$-confidence set, 
    use the 
    concentration properties of the data to control the 
    true positive rate.
\end{enumerate}

Note that we will in practise normalize $\eemp$ to have
unit diagonal in order to consistently recover the support.

\subsection{Concentration Confidence Set}
\label{sec:conConfSet}

The first step is to construct a confidence set for $\strue$ 
about $\eemp$.  Theoretical justification of the following 
is provided in Section~\ref{sec:falsePos}.

Given a false positive rate $0<\rho\le0.5$, we construct a 
ball $B_\rho$ centred on $\eemp$ as follows.  
\begin{enumerate}
  \item 
    Find $\eta = 2^a\rho \in(0.5,1]$ for some $a\in\integer^+$.
  \item 
    Compute $\lmb$, the $\eta$-quantile of the magnitudes of the 
    off-diagonal entries in $\eemp$.  That is,  $\lmb>0$ is 
    the smallest real number such that
    $$
        \frac{
          \#\left\{\abs{\hat{\sigma}_{i,j}}>\lmb\,|\,i<j\right\}
        }{d(d-1)/2}\le\eta.
    $$
  \item 
    Apply hard thresholding to $\eemp$ with threshold $\lmb$ to 
    get $\eemp_\lmb$ whose entries are 
    $$
      (\eemp_\lmb)_{i,j} = \left\{\begin{array}{ll}
        \hat{\sigma}_{i,j} & i=j\text{ or }
        \abs{\hat{\sigma}_{i,j}}\ge\lmb\\
        0 & \text{otherwise}
      \end{array}
      \right.
    $$
    which is, set off-diagonal entries to zero if they were originally
    less than $\lmb$ in magnitude.
  \item
    Construct the operator norm ball about $\eemp$ of radius 
    $r = 2^a\norm{\eemp-\eemp_\lmb}_\infty$.
  \item 
    Use a suitable concentration inequality to determine a bound on 
    the coverage of this ball as a confidence set.
\end{enumerate}

%To construct such a confidence set about $\eemp$, 
%concentration inequalities are employed.  
%Specific inequalities are chosen 
%based on data assumptions and are discussed in subsequent 
%Sections~\ref{sec:logconmeas} %, \ref{sec:boundedrv}, 
%and~\ref{sec:subexponential}.
%In general, the inequalities all take a similar form.  
%This $r_\alpha$ will be referred to as the \textit{deviation threshold}.
%Table~\ref{tab:specAssump} contains some explicit choices for 
%the metric and deviation threshold given specific assumptions 
%on the $X_i$, which are used to drive the choice of concentration
%inequality.  Log concave and sub-exponential measures 
%are discussed separately in
%Section~\ref{sec:sparseEstCases}.  
%Our method did not perform 
%well in the case of bounded random variables with details on 
%such in Appendix~\ref{sec:boundedrv}.

\iffalse
\begin{table}
  \begin{center}
  \begin{tabular}{llc}
    \hline
    Assumption on $X_i$ & \multicolumn{1}{l}{~$d(\espar,\strue)$} & 
    $r_\alpha$ \\
    \hline
    Log Concave Measure  & $\norm*{\espar-\strue}_p^{1/2}$ &
    $\sqrt{(-2/c_0n)\log\alpha}$\\
    Bounded in norm  & $\norm*{\espar-\strue}_p$ &
    $U\sqrt{(-1/2n)\log\alpha}$\\
    Sub-Exponential Measure & $\norm*{\espar-\strue}_p^{1/2}$ &
    $\max\{ -K\log\alpha/\sqrt{n}, \sqrt{-K\log\alpha} \}$
    \\\hline
  \end{tabular}
  \end{center}
  \capt{ 
    \label{tab:specAssump}
    A few metrics $d(\cdot,\cdot)$ and radii
    $r(\alpha)$ given specific assumptions on the data $X_i$
    discussed in Section~\ref{sec:sparseEstCases}.
  }
\end{table}
\fi

What we have now is
$$
  B_\rho = \left\{
    \Pi\in\real^{d\times d} : 
    \norm{ \Pi - \eemp }_\infty \le r
  \right\}.
$$
This set will be searched for its sparsest member using the
algorithm in the following subsection.

%To actually identify such a $\espar$, 
%we require some criterion to optimize over all elements of
%the confidence set.  
%Section~\ref{sec:zeroMeth} takes inspiration from thresholding
%techniques for sparse covariance estimation
%\citep{BICKELLEVINA2008,ROTHMAN2009,CAILUI2011} 
%and implements them inside this concentration framework.
%It begins with $\eemp$ and
%attempts to threshold it as much as possible 
%while still remaining in the confidence set.
%The Procrustes method in Section~\ref{sec:procMeth}
%is more closely related to the shrinkage estimators
%\citep{DANIELS1999,DANIELS2001,HOFF2009,JOHNSTONE2012}.
%It chooses $\espar$ to be a convex combination of 
%$\eemp$ and some sparse target matrix using the Procrustes 
%size and shape distance, which has been shown to be a 
%useful metric when one is concerned with inference 
%in the space of covariance matrices \citep{DRYDEN2009}.

\subsection{Thresholding within confidence sets}
\label{sec:zeroMeth}

A generalized thresholding operator, as defined in 
\cite{ROTHMAN2009}, is
$s_\lmb(\cdot):\real\rightarrow\real$ 
such that
$$
  \abs{s_\lmb(z)}\le z,~~
  s_\lmb(z)=0\text{ for }\abs{z}\le\lmb,~~\text{and }
  \abs{s_\lmb(z)-z}\le\lmb,
$$
which will apply element-wise to a matrix.
In the past, such an operator is applied to the 
empirical estimate $\eemp$ for some $\lmb$ generally
chosen via cross validation.
Instead of directly choosing a threshold $\lmb$,
our approach is to find the largest $\lmb$ such that
$d(s_\lmb(\eemp),\eemp) \le r$.
\begin{enumerate}
  \item Set $\espar_0 = (\ediag)^{-1/2}(\eemp)(\ediag)^{-1/2}$
    to be the empirical estimator normalized to have a 
    diagonal of ones.
    Initialize the threshold to $\lmb=0.5$ and 
    write $\espar_\lmb = s_\lmb(\eemp)$.  Let $k=1$
    be the number of steps of the recursion.
    Choose a false positive rate $\rho$ and compute $r$
    as in the previous section. 
  \item 
    Increase $k \leftarrow k+1$, then update $\lmb$ as follows.
    \begin{enumerate}
      \item if $d(\espar_\lmb,\eemp) \le r$,
            set $\lmb \leftarrow \lmb + 2^{-k-1}$.
      \item Otherwise, set $\lmb \leftarrow \lmb - 2^{-k-1}$.
    \end{enumerate}
  \item 
    Repeat step ii until $k$ has reached the desired number of
    iterations.  Generally, as few as $k=10$ will suffice.
  \item The resulting estimator is
    $\espar = (\ediag)^{1/2}(\espar_\lmb)(\ediag)^{1/2}$ 
    where $\espar_\lmb$ is the final matrix resulting from
    this recursion.  
\end{enumerate}

\begin{remark}[Positive Definite Estimators]
    If $\espar$ is not positive semi-definite,
    then it can be projected onto the space of positive semi-definite 
    matrices.  A standard past approach is to map the negative 
    eigenvalues to zero or to their absolute value, which maintains
    the eigen-structure.
  However, such a projection will have 
  an adverse effect on the support recovery problem as the estimator
  will no longer be sparse. 
  An alternative is to map  $\espar \rightarrow \espar + \gamma I_d$
  for some $\gamma>0$ large enough to make the result positive
  definite.  This will not effect the recovered support of the matrix.
  More clever projections may also be possible.
\end{remark}

In the case that the metric $d(\cdot,\cdot)$ is a monotonically increasing
function of the Hilbert-Schmidt / Frobenius norm $\norm{\espar_\lmb-\eemp}_2$
or another entrywise norm,
then the sequence $d(\espar_\lmb,\eemp)$ will be increasing in $\lmb$.
\begin{proposition}
  In the context of the above algorithm, if $\lmb_1>\lmb_2$, then
  for any $p,q$, we have
  $$
    \norm{\espar_{\lmb_1}-\eemp}_{p,q} \ge 
    \norm{\espar_{\lmb_2}-\eemp}_{p,q}.
  $$
\end{proposition}
\begin{proof}
  As $\lmb_1>\lmb_2$, the entries of the matrix 
  $\espar_{\lmb_1}-\eemp$ are equal to or larger in absolute value than 
  the entries of $\espar_{\lmb_2}-\eemp$.  Hence 
  $
    \norm{\espar_{\lmb_1}-\eemp}_{p,q} \ge 
    \norm{\espar_{\lmb_2}-\eemp}_{p,q}
  $
  by definition~\ref{defn:entrywiseNorm}.
\end{proof}
%which follows quickly from that norm being equivalently 
%expressed as the sum of the squared entries in a matrix.
This property guarantees that the above algorithm will find the 
sparsest $\espar$ in the confidence set 
in the sense of having the largest threshold possible. 
However, for an arbitrary metric or specifically other 
$p$-Schatten norms, this sequence 
may not necessarily be strictly increasing in $\lmb$.  
Another commonly used norm, which will be shown in 
Section~\ref{sec:numeric} to give superior performance in simulation, 
is the operator norm 
$\norm{\espar_\lmb-\eemp}_\infty$, which does not yield a monotonically 
increasing sequence.  
Though,  
this sequence is roughly increasing in the sense that it 
is lower bounded by definition by the maximum $\ell^2$ norm of the columns of 
$\espar_\lmb-\eemp$, which is an increasing sequence.  
Furthermore, it is upper bounded by the $\ell^1$ norm of 
the columns of $\espar_\lmb-\eemp$, which follows from 
the Gershgorin circle theorem
\citep{ISERLES2009}, and which is also an increasing sequence.
%In practice, the operator norm in particular gives superior 
%performance in the numerical simulations of Section~\ref{sec:numeric}.

\section{Fixing a false positive rate}
\label{sec:falsePos}

For many sparse matrix estimation methods, 
theorems demonstrating sparsistency are proved.
These indicate that in some asymptotic sense, the
correct support of the true matrix will
eventually be recovered generally as $n$ and $d$ grow 
together at some rate.
However, none provide a method for fixing a false positive
rate and finding an estimator that satisfies such a rate,
which is certainly of interest to any practitioner with
a finite fixed sample size.  Hence, we present a method
for tuning our parameter $\alpha$ to a desired false 
positive rate for the covariance estimator.  

Before proceeding, we will require a class of sparse matrices
similar to those from 
\cite{BICKELLEVINA2008,BICKELLEVINA2008A,ROTHMAN2009,CAILUI2011}.
Specifically, let 
%\begin{multline}
$$
  \label{eqn:sparseClass}
  \mathcal{U}(\kappa,\delta)=\left\{
    \Sigma\in\real^{d\times d}:
    \max_{i=1,\ldots,d}\sum_{j=1}^d \Indc{\sigma_{i,j}\ne0} \le \kappa,
    \text{ if } \sigma_{i,j}\ne0 \text{ then }
    \abs{\sigma_{i,j}}\ge\delta > 0
  \right\}.
$$
%\end{multline}
For the results regarding the false positive rate, we are 
not concerned with the lower bound $\delta$ and only with 
$\kappa$, the maximum number of non-zero entries per column or row.
As long as $\kappa$ increases more slowly than the dimension $d$,
which is made specific below, we can achieve a desired false
positive rate without interference.

For an estimator $\tilde{\Sigma}\in\real^{d\times d}$, 
the false positive rate is 
$$
  \rho(\tilde{\Sigma}) = 
  \frac{
    \#\{\tilde{\sigma}_{i,j}\ne0\,|\,\sigma_{i,j}=0,\,i>j\}
  }{
    d(d-1)/2
  }
$$
where $\sigma_{i,j}$ is the $ij$th entry of the true covariance 
matrix and $\tilde{\sigma}_{i,j}$ is the $ij$th entry of the 
estimator $\tilde{\Sigma}$.  Hence, we are counting the number
of non-zero entries in our estimator that should have been zero.
For notation, let $\eemp$ be the usual empirical estimate 
of the covariance matrix.  Let $\eemp_{0}$ be the empirical
estimator with all off diagonal entries set to zero thus 
guaranteeing a false positive rate of zero.  For $\eta\ge0.5$, 
let $\eemp_{\eta}$ be the empirical estimator after application
of the strong threshold operator with threshold
$M_{\eta}=\text{quantile}(\abs{\hat{\sigma}_{i,j}},\eta\,:\,i>j)$, which
removes $100(1-\eta)$\% of the off diagonal entries achieving a false 
positive rate of approximately $(1-\eta)$ due to the following lemma.
\begin{lemma}
  \label{lem:quant}
  Let $\strue\in\mathcal{U}(\kappa,\delta)$ from 
  Equation~\ref{eqn:sparseClass} with $\kappa = o(d^\nu)$.
  %and the random vectors $X_i\in\real^d$ such that for some $q\ge2$,
  %$\xv \abs{X_{i,j}}^q<C_q<\infty$ for $i=1,\ldots,n$, 
  %$j=1,\ldots,d$.
  Let the $\eta\in[0.5,1)$ threshold, $M_{\eta}$, be the $\eta$
  quantile of $\abs{\hat{\sigma}_{i,j}}$ with $i>j$,
  and let the corresponding thresholded estimator be
  $\eemp_{\eta} = s_{M_{\eta}}(\eemp)$ with $ij$th entry denoted
  $\hat{\sigma}^{(\eta)}_{i,j}$.  Then, denoting 
  $
    \hat{\eta} = 
      \#\{(i,j) \,|\, i>j,\abs{\hat{\sigma}^{(\eta)}_{i,j}}>0,\sigma_{i,j}=0\}
    [d(d-1)/2]^{-1}
  $, we have that for $\veps>0$
  $$
      \abs{\hat{\eta}-\eta} \le Cd^{\nu-1}.
  $$
  for some constant $C>0$.
\end{lemma}  
\begin{remark}
  For this lemma, we want
  the $\eta$-quantile of the mean zero entries, but have
  to work with the $\eta$-quantile of the entire collection,
  which is contaminated by a small number of elements with non-zero
  mean.
  For $\nu<1$, the error is $O(d^{\nu-1})$ hence for 
  $\eta\approx0.5$,
  thresholding based on the $\eta$-quantile suffices for large
  enough $d$.  
  For small $\eta$, say $\eta \approx d^{-1}$, we have to 
  work harder motivating Theorem~\ref{thm:falsepos} below.
\end{remark}

As noted in the remark, we cannot continue to threshold 
based on the sample quantiles for very small false positive
rates.
However,
using the matrices, $\eemp_{\eta}$ and $\eemp_0$, as 
reference points, we can interpolate via the following 
theorem to achieve any desired
false positive rate.
\begin{theorem}
  \label{thm:falsepos}
  Let $\strue\in\mathcal{U}(\kappa,\delta)$ from 
  Equation~\ref{eqn:sparseClass}
  with $\kappa=O(d^{\nu})$ for $\nu<1/2$.
  Given a desired false positive rate, $\rho\in(0,0.5]$,
  and $\eta=\rho2^{a}\in(0.5,1]$ for some $a\in\integer^+$, let  
  $\eemp_\rho$ be the hard thresholded empirical estimator
  that achieves a false positive rate of $\rho$. 
  Then, 
  $$
    \abs*{\eta\frac{ 
      \xv\norm{\eemp_\rho-\eemp_0}_\infty 
    }{
      \xv\norm{\eemp_{\eta}-\eemp_0}_\infty
    } - \rho}
    \le
    K_1n\rho^{1/2}d^{-1/4} +
    K_2n\rho^{1/4}d^{-1/2} + o(nd^{-1/2})
  $$
  where $K_1,K_2$ are universal constants.
\end{theorem}
\begin{remark}
  The above
  Theorem~\ref{thm:falsepos} is wholly uninteresting for 
  large values of $n$.  However, its power arises in the
  non-asymptotic realm of interest---namely 
  when $d\gg n$---and also from highlighting the interplay
  between the dimension, sample size, and $\rho$, the 
  sparseness of the estimator.
  Furthermore, this result does not require any distributional
  assumption.  It also does not require any assumption on the
  lower bound $\delta$ on the non-zero $\abs{\sigma_{i,j}}$
  as it is only concerned with the $\sigma_{i,j}$ that are zero.
\end{remark}

The proof of the above theorem relies on the following lemma
involving symmetrization of random covariance matrices, 
which may be of independent interest.
\begin{lemma}
  \label{lem:covSymmet}
  Let $R\in\real^{d\times d}$ be a real valued symmetric
  random matrix with zero diagonal and mean zero entries 
  bounded by 1 and
  not necessarily iid, and 
  let $B\in\{0,1\}^{d\times d}$ 
  be an iid symmetric Bernoulli random matrix with entries
  $b_{i,j}=b_{j,i}\dist\distBern{\rho}$ for 
  $\rho\in(0,1)$.  
  Denoting the entrywise or Hadamard product by $\circ$,
  let $A=R\circ B$.
  %of these off-diagonal entries are non-zero.  
  Let $\Veps\in\{-1,1\}^{d\times d}$ be a symmetric
  random matrix with iid Rademacher entries $\veps_{i,j}$ 
  for $j<i$ and $\veps_{i,j}=\veps_{j,i}$.  Then, 
  %denoting the entrywise or Hadamard product by $\circ$,
  $$
    \xv\norm{
      A\circ\Veps
    }_\infty \le
      K_1d^{1/2}\rho^{1/4} + K_2d^{3/4}\rho^{1/2}
  $$
  where $K_1,K_2$ are universal constants.
\end{lemma}

\section{Concentration Confidence Sets}
\label{sec:sparseEstCases}

The following three subsections detail different assumptions on 
the data under scrutiny and the specific concentration
results that apply in these cases.  We consider sub-Gaussian 
concentration for log concave measures and for bounded random
variables.
We also consider sub-exponential concentration.  However, this 
collection is by no means exhaustive.  Given the wide variety 
of concentration inequalities being developed, our approach 
can be applied much more widely than to merely these three
settings.

Let $d(\cdot,\cdot)$ be some metric measuring the distance
between two covariance matrices, and let 
$\psi:\real\rightarrow\real$ be monotonically increasing.  Then,
the general form of the concentration inequalities is
$$
  \prob{ 
    d(\strue,\eemp) \ge 
    \xv d(\strue,\eemp) + r
  } \le
  \ee^{-\psi(r)},
$$
which is a bound on the tail of the distribution of 
$d(\strue,\eemp)$ as it deviates above its mean.
Thus, to construct a $(1-\alpha)$-confidence set, the variable
$r = r_\alpha$ is chosen such that $\exp(-\psi(r_\alpha))=\alpha$.

Now, let $\espar$ be our sparse estimator for $\strue$.
We want these two to be close in the sense of the above confidence set
and therefore choose a
$\espar$ such that $d(\espar,\eemp) \le r_\alpha$.
Consequently, we have that
\begin{align*}
  &\prob{
    d(\espar,\strue) \ge 
    \xv d(\eemp,\strue) +2r_\alpha
  }\\
  &~~~~\le \prob{
    d(\espar,\eemp) + d(\eemp,\strue) \ge 
    \xv d(\eemp,\strue) +2r_\alpha
  } \\
  &~~~~\le \prob{
    d(\eemp,\strue) \ge 
    \xv d(\eemp,\strue) + r_\alpha 
  } %\\
  %&~~~~
  \le \exp(-\psi(r_\alpha)) = \alpha.
\end{align*}
Hence, we choose $\espar$ close enough to $\eemp$ to 
share its elegant concentration properties, but far enough
away to result in a better estimator for $\strue$.

\subsection{Log Concave Measures}
\label{sec:logconmeas}

In this section, the general methods from Section~\ref{sec:sparseEstProc}
are specialized for an iid sample $X_1,\ldots,X_n\in\real^d$
whose common measure $\mu$ is \textit{strongly log-concave}.
This property implies dimension-free sub-Gaussian concentration
and includes such common distributions as 
the multivariate Gaussian, Chi, and Dirichlet distributions.

\begin{definition}[Strongly log-concave measure]
  \label{def:logConcave}
  A measure $\mu$ on $\real^d$ is strongly log-concave if
  there exists a $c>0$ such that $d\mu = \ee^{-U(x)}dx$ and
  $\text{Hess}(U)-c I_d \ge 0$ (i.e. is non-negative definite)
  where $\text{Hess}(U)$ is the $d\times d$ matrix of second derivatives.
\end{definition}

  From Corollary~S4.5%\ref{thm:prodconc},
  let $X_1,\ldots,X_n\in\real^d$ have measures 
  $\mu_1,\ldots,\mu_n$, which are all strongly log-concave
  with coefficients $c_1,\ldots,c_n$.  Let 
  $\nu = \mu_1\otimes\ldots\otimes\mu_n$ be the product measure
  on $\real^{d\times n}$.  Then,
  for any $1$-Lipschitz $\phi:(\real^d)^n\rightarrow\real$
  and for any $r>0$,
  $$ 
    \prob{\phi(X_1,\ldots,X_n)\ge \xv\phi(X_1,\ldots,X_n)+r}
    \le \ee^{-\min_{i}c_ir^2/2}.
  $$
This follows from 
Theorem~S4.4%\ref{thm:prodmeasure} 
and the other results 
contained within the supplementary material.%Appendix~\ref{app:logConcave}.
For a detailed exposition of how sub-Gaussian concentration 
is established for log concave measures, see Chapter 5 of
\cite{LEDOUX2001}.   Examples include the multivariate Gaussian
and the Dirichlet distributions.

To make use of the above result, we must choose
a suitable Lipschitz function $\phi(\cdot)$.
Let $X_1,\ldots,X_n,X\in\real^d$ be \iid random variables
with covariance $\strue$ and 
with a common strongly log-concave measure $\mu$
with coefficient $c>0$.  Let $\lmb_1\ge\ldots\ge\lmb_n$
be the eigenvalues of $\Sigma$ and $\Lambda=(\lmb_1,\ldots,\lmb_n)$.
For some $p\in[1,\infty]$, let $\norm{\cdot}_p$ be the $p$-Schatten
norm, which in this case is $\norm{\Sigma}_p=\norm{\Lambda}_{\ell^p}$.
Note that $\norm{X\TT{X}}_p = \norm{X}_{\ell^2}^2$ for any
$p\in[1,\infty]$.
Define the function $\phi$ to be
$
  \phi(X_1,\ldots,X_n) = 
  \norm*{ 
    \frac{1}{n}\sum_{i=1}^n (X_i-\xv X)\TT{(X_i-\xv X)}
  }_p^{1/2}.
$

For $p\in\{1,2,\infty\}$, we have that
$\phi$ is Lipschitz with coefficient 
$\norm{\phi}_\mathrm{Lip}=n^{-1/2}$ with respect to
the Frobenius or Hilbert-Schmidt metric, which is established
in Proposition~S3.5%\ref{prop:lipDp2} 
for
$p=2$ and $p=\infty$ and in Proposition~S3.2%\ref{prop:lipDp1}
for $p=1$.
That is, let $X_1,\ldots,X_n,Y_1,\ldots,Y_n\in\real^d$, and denote 
${\bf X}=(X_1,\ldots,X_n)$ and ${\bf Y}=(Y_1,\ldots,Y_n)$,
then
$
  \abs{\phi({\bf X})-\phi({\bf Y})} \le
  n^{-1/2} d_{2,2}({\bf X},{\bf Y}) 
  = \left( \frac{1}{n}\sum_{i=1}^n\norm{X_i-Y_i}_{\ell^2}^2 \right)^{1/2}.
$
%This is established in Proposition~\ref{prop:lipDp2} for
%$p=2$ and $p=\infty$.
%Proposition~\ref{prop:lipDp1} establishes that $\phi(\cdot)$
%is also Lipschitz with coefficient $\norm{\phi}_\mathrm{Lip}=n^{-1/2}$
%for $p=1$.
From here, the procedure outlined in Section~\ref{sec:sparseEstProc}
can be considered with the given $\phi$ and 
$r_\alpha = \sqrt{(-2/nc_0)\log\alpha}$.  

In many cases, including the two examples above, the 
constructed confidence set is completely dimension-free.
Thus, even mild assumptions on the relationship between
the sample size $n$ and the dimension $d$, such as 
$\log d = o(n^{1/3})$ from the adaptive soft thresholding estimator
of \cite{CAILUI2011}, are not needed to prove consistency 
in our setting.
Furthermore, the concentration inequalities immediately 
give us a  fast rate of convergence 
as long as $-\log\alpha = o(n)$ with a proof provided 
in the supplementary material.%Appendix~\ref{sec:proofs}.
\begin{theorem}
  \label{thm:logConConv}
  Let $X_1,\ldots,X_n\in\real^d$ be iid with common
  measure $\mu$.  Let $\mu$ be strictly log concave 
  with some fixed constant $c_0$ from Definition~\ref{def:logConcave}.
  Then, for $\alpha\in(0,1)$, $p\in[1,\infty]$, and
  $r_\alpha = \sqrt{(-2/nc_0)\log\alpha}$,
  $$
    \sup_{\espar:\norm*{\espar-\eemp}_p\le r_\alpha}
    \prob{  
      \norm*{\espar-\strue}_p \ge
      O\left(n^{-1/2}( 1 + n^{-1/4}\sqrt{-\log\alpha} )^2\right)
    } \le \alpha.
  $$
\end{theorem}

\begin{remark}
  This theorem effectively says that choosing an estimator
  in the ball centred around $\eemp$ cannot be too bad 
  assuming the niceness of log-concave measures.  It also
  tells us how fast we can shrink the ball as $n$ increases.
\end{remark}

A second and arguably more important issue, 
see the supplementary material,%Appendix~\ref{app:empDiag},
in the setting of sparse covariance estimation is
that of support recovery or ``sparsistency'' 
\citep{LAMFAN2009,ROTHMAN2009}.  To recover the support of a 
covariance matrix---that is, determine which entries 
$\sigma_{i,j}\ne0$---we will require a class of sparse matrices
  from Equation~\ref{eqn:sparseClass}.
In past work, a notation of ``approximate sparsity'' is considered
where the first condition in $\mathcal{U}(\kappa,\delta)$ is replaced
with $\max_{i=1,\ldots,d}\sum_{i=1}^d \abs{\sigma_{i,j}}^q < \kappa$
for $q\in[0,1)$.  However, once we bound the non-zero entries 
away from zero by some $\delta$, such ``approximate sparsity''
implies standard sparsity with $q=0$.  It is worth noting that
the above Proposition~\ref{thm:logConConv} does not require 
such a sparsity class, because our estimator is forced to remain 
close enough to $\eemp$ to follow $\eemp$'s convergence to $\strue$.

\begin{theorem}
  \label{thm:logConSpar}
  Let $X_1,\ldots,X_n\in\real^d$ be iid with common
  measure $\mu$.  Let $\mu$ be strictly log concave 
  with some fixed constant $c_0$ from Definition~\ref{def:logConcave}.
  Furthermore, let $\strue \in \mathcal{U}(\kappa,\delta)$
  and let $\delta = O(n^{-1+\veps})$ for any $\veps>0$.
  Then, for $\espar$ denoting the concentration estimator 
  using the hard thresholding estimation from Section~\ref{sec:zeroMeth}
  with the operator norm metric,
  $$
    \lim_{n\rightarrow\infty}
    \prob{ \mathrm{supp}(\espar) \ne \mathrm{supp}(\strue)  } = 0
  $$
  where $\mathrm{supp}(\Sigma) = \{(i,j): \sigma_{i,j}\ne0\}$.
\end{theorem}

\begin{remark}
Note that the condition that $\delta = O(n^{-1+\veps})$ allows
for a much quicker decay of the non-zero entries of $\strue$
than in \cite{KAROUI2008} where the lower bound is of the 
form $Cn^{-\alpha_0}$ with $0<\alpha_0<1/2$.  It is also 
much quicker than  the similar rate achieved in
\cite{ROTHMAN2009} where the lower bound is any $\tau$ such 
that $\sqrt{n}\tau$ increases faster than $\sqrt{\log(d)}$ 
with the enforced asymptotic condition that $\log(d)/n = o(1)$
resulting in a rate no faster than $n^{-1/2}$. 
Though, it is worth noting that if $\delta$ decays to zero at
a faster rate, then the above convergence rate for support recovery 
slows as can be seen in the proof.
\end{remark}

\section{Numerical simulations}
\label{sec:numeric}

In the following subsections, we apply the methods from the 
previous sections to three multivariate distributions of interest:
the Gaussian, Laplace, and Rademacher distributions.
In doing so, we apply Theorem~\ref{thm:falsepos} to analytically
determine the ideal confidence ball radius in order to construct
a sparse estimator of $\strue$.  We also compare the support 
recovery of our approach against penalized estimators and 
standard application of universal threshold estimators.

As mentioned before, our proposed concentration confidence set
based method has a similar feel to regularized / penalized estimators
as the larger the constructed confidence set is, the sparser the
returned estimator will be.
Thus, we compare our approach with the following lasso style 
estimator
from the R package \texttt{PDSCE} \citep{PDSCE},
which optimizes 
$$
  \hat{\Sigma}^\text{PDS} = 
  \argmin{\Sigma\ge0}\left\{
    \norm{\Sigma - \eemp}_2 - \tau \log \det( \Sigma )  +
    \lambda\norm{\Sigma}_{\ell^1}
  \right\}
$$
with $\tau,\lmb>0$.
Here, the $\log \det$ term is used to enforce positive 
definiteness of the final solution, and $\norm{\cdot}_{\ell^1}$
is the lasso style penalty, which enforces sparsity.

The similar method from the R package \texttt{spcov} \citep{SPCOV},
which
uses a majorize-minimize algorithm to determine 
$$
  \hat{\Sigma}^\text{MMA} = 
  \argmin{\Sigma\ge0}\left\{
    \tr{ \eemp \Sigma^{-1} } - \log \det(\Sigma^{-1}) + 
    \lambda\norm{\Sigma}_{\ell^1}
  \right\}
$$
for some penalization $\lmb>0$, was also considered but
proved to run too slowly on high dimensional 
matrices---that is, $d\ge200$---to be included in the numerical experiments.

Of course, we also compare our method against the four universal thresholding
estimators applied to the empirical covariance 
matrix from \citep{ROTHMAN2009},
Hard, Soft, SCAD, and Adaptive LASSO:
%Such estimators are constructed by applying a thresholding function,
%which satisfies some nice properties, to the empirical covariance
%matrix.  
%The two types of thresholding considered in our numerical
%experiments will be \textit{Hard}, which zeros any entries smaller
%than some $\lmb>0$,
%and \textit{Soft}, which shrinks the entries by some $\lmb>0$.
\begin{align*}
  &\hat{\Sigma}^\text{Hard}_\lmb
  &&= \{ \hat{\sigma}_{i,j}\Indc{\hat{\sigma}_{i,j}>\lmb} \}_{i,j}\\
  &\hat{\Sigma}^\text{SCAD}_\lmb
  &&= \left\{\begin{array}{ll}
    \hat{\sigma}_{i,j}^\text{Soft} & 
    \text{ for }\hat{\sigma}_{i,j}\le2\lmb\\
    \frac{a-1}{a-2}(\hat{\sigma}_{i,j}-2\lmb)+\lmb &
    \text{ for } 2\lmb < \hat{\sigma}_{i,j}\le a\lmb\\
    \hat{\sigma}_{i,j}^\text{Hard} & 
    \text{ for }\hat{\sigma}_{i,j}>a\lmb
  \end{array}\right. \\
  &\hat{\Sigma}^\text{Soft}_\lmb
  &&= \{ 
    \text{sign}(\hat{\sigma}_{i,j})(\abs{\hat{\sigma}_{i,j}}-\lmb)_+
  \}_{i,j} \\
  &\hat{\Sigma}^\text{Adpt}_\lmb
  &&= \{ 
    \text{sign}(\hat{\sigma}_{i,j})(
      \abs{\hat{\sigma}_{i,j}}-\lmb^{\eta+1}\abs{\hat{\sigma}_{i,j}}^{-\eta}
    )_+
  \}_{i,j}\\
\end{align*} 
where $\hat{\sigma}_{i,j}$ is the $(i,j)$th entry of the empirical 
covariance estimate, $a=3.7$, and $\eta=1$. The parameter $\lmb>0$ 
is the threshold, 
which is chosen in practice via cross validation with respect to
the Hilbert-Schmidt norm.
Briefly,
the data is split in half, two empirical estimators are formed,
one is thresholded, and $\lmb$ is selected to minimize the
Hilbert-Schmidt distance
between the one empirical estimate and the other thresholded estimate.

\subsection{Multivariate Gaussian Data}
\label{sec:simGaussian}

Let $X_1,\ldots,X_n\in\real^d$ be \iid mean zero random vectors
with a strictly log concave measure and covariance matrix $\Sigma$.
By Corollary~S4.5,%\ref{thm:prodconc}, 
there exists a constant
$c_0>0$ such that 
$
  \prob{ 
    \phi({\bf X}) \ge \xv \phi({\bf X}) + r
  } \le
  \ee^{-nr^2/2c_0}
$
where 
$
  \phi({\bf X}) = \norm{\eemp - \Sigma}_p^{1/2}
$ where $\eemp=n^{-1}\sum_{i=1}^n (X_i-\bar{X})\TT{(X_i-\bar{X})}$ is the
empirical estimate of the covariance matrix. This
results in the size $1-\alpha$ confidence set for $\Sigma$
$$
  \mathcal{C}_{1-\alpha} = \left\{
    \Sigma: 
    \norm{\eemp - \Sigma}_p^{1/2} \le 
    \xv\norm{\eemp - \Sigma}_p^{1/2} +
    \sqrt{ (-2c_0/n) \log \alpha }
  \right\}
$$
for $\alpha \in (0,1)$.  In the notation of 
Section~\ref{sec:logconmeas}, 
$r_\alpha = \sqrt{ (-2c_0/n) \log \alpha }$.

In the multivariate Gaussian case, $c_0$ is the maximal 
eigenvalue of the covariance matrix $\Sigma$.
As mentioned before,
we avoid any issues of estimating $c_0$ in practice.
Regardless of our choice for $c_0$ tuning the regularization
parameter $\alpha$ to a specific false positive rate 
negates the need for an accurate estimate of $c_0$.  

Table~\ref{tab:fptpGauss} displays false positive and 
true positive percentages for seven sparse estimators 
computed over 100 replications
of a random sample of size $n=50$ of $d=50,100,200,500$
dimensional 
multivariate Gaussian 
data with a tri-diagonal covariance matrix $\Sigma$ whose
diagonal entries are 1 and whose off-diagonal entries are 0.3.
We can clearly see that the concentration-based estimator  
approaches the desired false positive rate---either 1\% or 5\%---as
the dimension increases.
In contrast, the thresholding estimators with 
threshold $\lmb$ chosen via cross validation generally start with 
higher false positive percentages, which tend to zero 
as the dimension increases.  As noted in previous work, hard 
thresholding is overly aggressive.   
The PDS method is very stable across changes in the dimension
and maintains a constant 3.4\% false positive rate and 50\%
true positive rate.

\begin{table}
  \begin{center}
  \begin{tabular}{l rrrr c rrrr}
       \hline
        & \multicolumn{4}{c}{False Positive \%} &&  
        \multicolumn{4}{c}{True Positive \%}  \\
    Dimension & 50 & 100 & 200 & 500 && 50 & 100 & 200 & 500 \\
    \hline
    CoM 1\% & 0.0&0.1&0.3&1.0&& 0.0& 7.7&20.7&32.0 \\
    CoM 5\% & 1.0&2.2&3.5&4.7&&33.1&42.9&51.5&56.0 \\
    PDS     & 3.4&3.4&3.4&3.4&&50.0&50.0&51.5&50.6 \\
    Hard    & 0.0&0.0&0.0&0.0&& 0.3& 0.0& 0.0& 0.0 \\
    Soft    & 2.0&0.7&0.2&0.0&&38.5&25.4&16.2& 7.5 \\
    SCAD    & 2.1&0.7&0.3&0.0&&39.0&26.0&16.4& 7.5 \\
    Adpt    & 0.3&0.1&0.0&0.0&&17.4&10.0& 5.8& 2.0 \\
       \hline
  \end{tabular}
  \end{center}
  \capt{
    \label{tab:fptpGauss}
    Percentage of false and true positives for multivariate Gaussian data
    and $\strue$ tri-diagonal with diagonal entries 1 and off-diagonal
    entries 0.3.
  }
\end{table}

\begin{figure}
  \begin{center}
    \includegraphics[width=0.475\textwidth]{\PICDIR/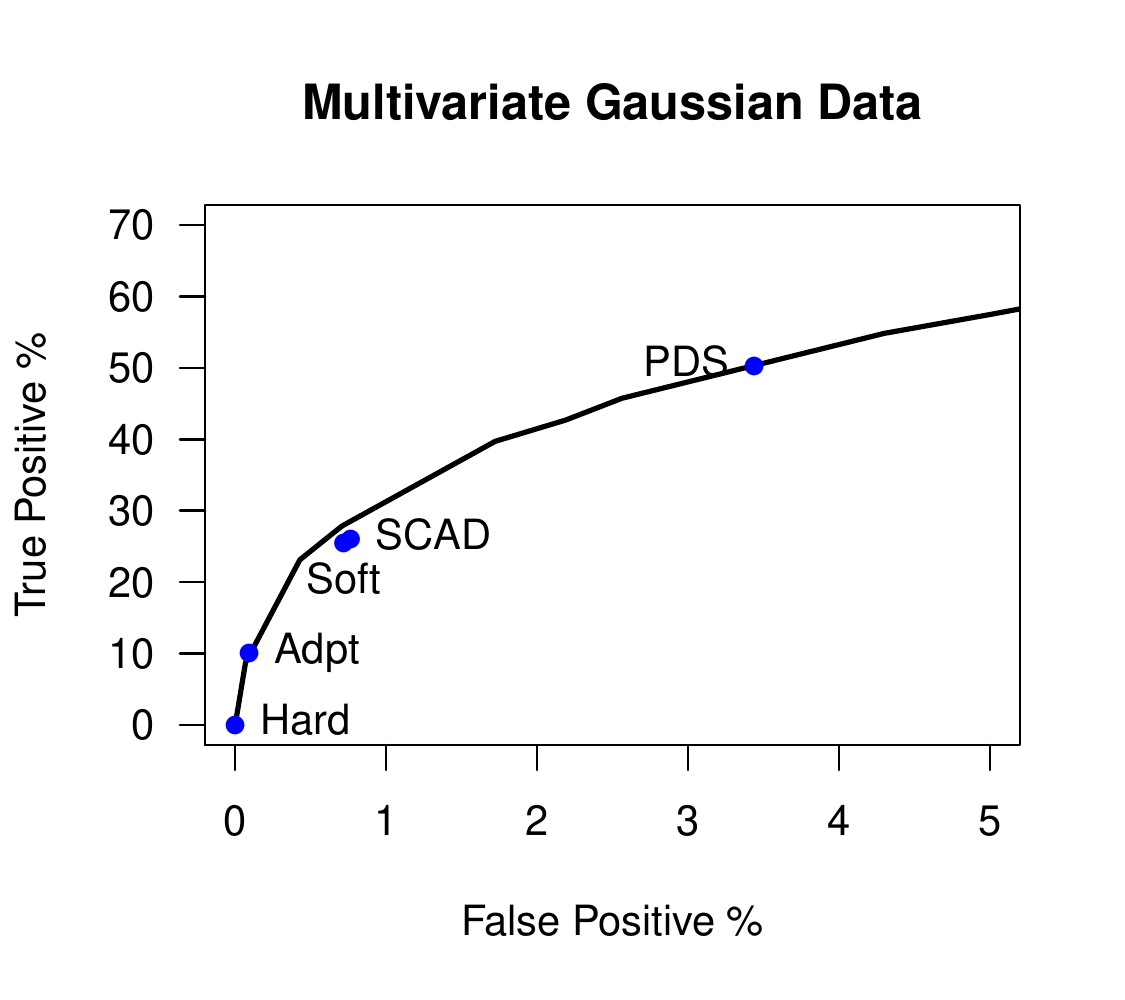}
    \includegraphics[width=0.475\textwidth]{\PICDIR/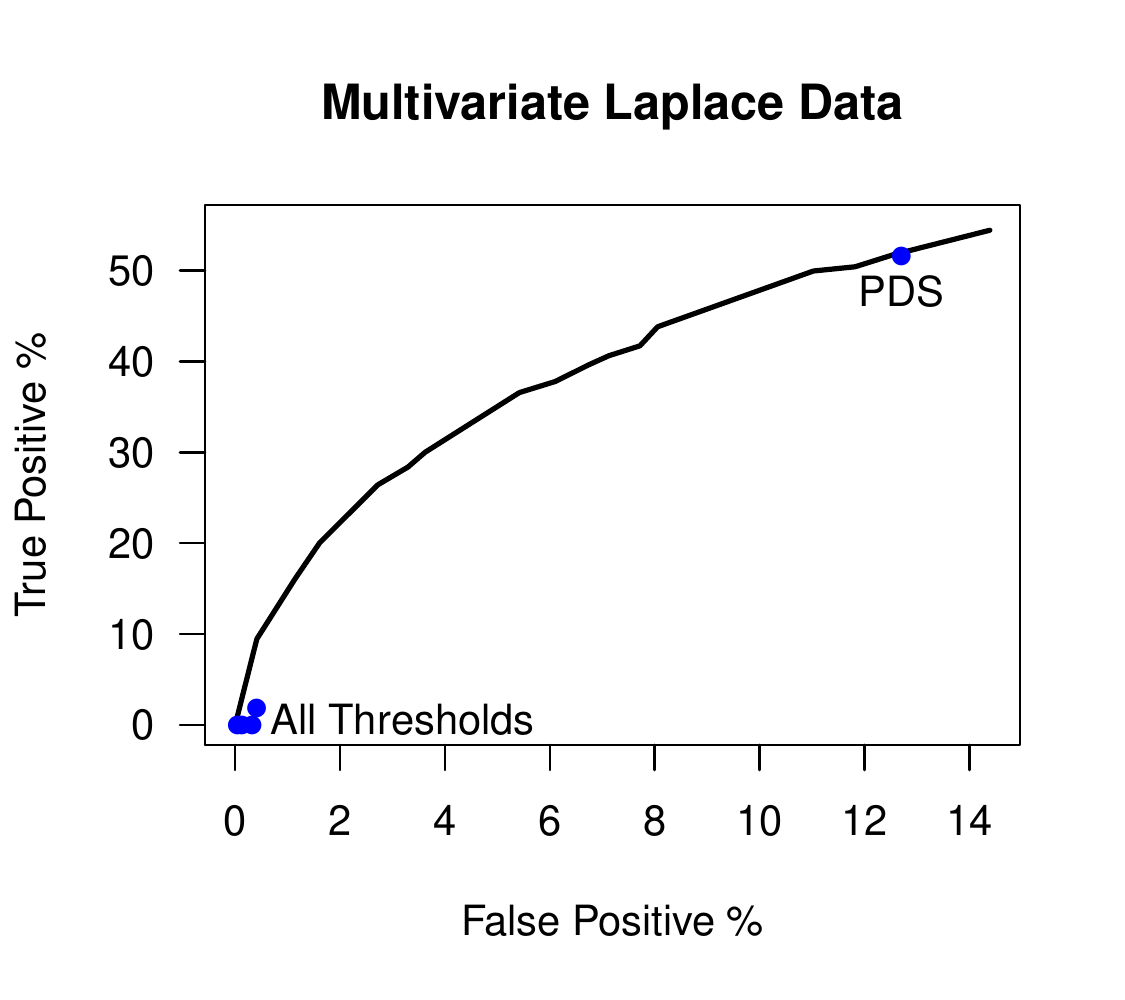}\\
    %\vspace{-0.2in}
    %\includegraphics[width=0.475\textwidth]{\PICDIR/fptp_rad_2018.pdf}
  \end{center}
  \capt{
    \label{fig:fptpBoth}
    A line demarcating the trade-off between false and true positive
    recoveries for multivariate Gaussian (left) and Laplace (right)
    %and Rademacher (bottom)
    data from 100 replications of sample size $n=50$ and 
    dimension $d=100$.  
    %The lines were formed by varying the false
    %positive rate for the concentration estimator.  
    %The blue dots indicate the performance of 
    %past methods for sparse covariance
    %estimation.
  }
\end{figure}

\subsection{Multivariate Laplace Data}

There are many possible ways to extend the univariate Laplace
distribution, also referred to as the double exponential distribution,
onto $\real^d$.  For the following simulation study,
we choose the extension detailed in \cite{ELTOFT2006}.
Namely, let $Z \dist \distNormal{0}{\sigma^2}$ and let
$V\dist\distExp{1}$.  Then, 
$X = \sqrt{V}Z \dist\distLaplace{\sigma/\sqrt{2}}$, which 
has pdf $f(x) = \sqrt{2}\sigma^{-1}\exp(-\sqrt{2}\abs{x}/\sigma )$
and variance $\var{X}=\sigma^2$.  For the multivariate setting,
now let $Z\in\real^d$ be multivariate Gaussian with zero mean 
and covariance $\strue$ and, once again, let $V\dist\distExp{1}$.
Then, we declare $X = \sqrt{V}Z$ to have a multivariate 
Laplace distribution with zero mean and covariance $\strue$.

Table~\ref{tab:fptpExp} displays false positive and 
true positive percentages for seven sparse estimators 
computed over 100 replications
of a random sample of size $n=50$ of $d=50,100,200,500$
dimensional 
multivariate Laplace
data with a tri-diagonal covariance matrix $\Sigma$ whose
diagonal entries are 1 and whose off-diagonal entries are 0.3.
Similarly to the previous setting, the concentration-based estimator  
approaches the desired false positive rate---either 1\% or 5\%---as
the dimension increases.
All universal thresholding estimators 
set most of the entries in the matrix to zero when 
threshold $\lmb$ chosen via cross validation.  
The PDS method
is still stable across changes in the dimension but fixates on
a much higher false positive rate around 12.5\% and a 
similar true positive rate of 51\%.

\begin{table}
  \begin{center}
  \begin{tabular}{l rrrr c rrrr}
    \hline
        & \multicolumn{4}{c}{False Positive \%} &&  
        \multicolumn{4}{c}{True Positive \%}  \\
    Dimension & 50 & 100 & 200 & 500 && 50 & 100 & 200 & 500 \\
    \hline
    CoM 1\% & 0.2& 0.4& 0.7& 1.1&&  4.5& 9.2&13.0&17.2\\
    CoM 5\% & 2.2& 3.3& 4.1& 4.7&& 22.8&29.3&32.1&34.1\\
    PDS     &12.4&12.7&12.2&12.2&& 51.0&51.5&51.0&51.2\\
    Hard    & 0.1& 0.0& 0.0& 0.0&&  0.0& 0.0& 0.0&0.0\\
    Soft    & 1.2& 0.4& 0.2& 0.0&& 11.3& 1.8& 0.0&0.0\\
    SCAD    & 0.8& 0.3& 0.2& 0.0&&  8.6& 0.0& 0.0&0.0\\
    Adpt    & 0.2& 0.1& 0.1& 0.0&&  0.0& 0.0& 0.0&0.0\\
    \hline
  \end{tabular}
  \end{center}
  \capt{
    \label{tab:fptpExp}
    Percentage of false and true positives for multivariate Laplace data
    and $\strue$ tri-diagonal with diagonal entries 1 and off-diagonal
    entries 0.3.
  }
\end{table}

%\subsection{Recovering Coefficient Matrices}

\subsection{Small Round Blue-Cell Tumour Data}

Following the same analysis performed in 
\cite{ROTHMAN2009} and subsequently in \cite{CAILUI2011},
we will consider the data set resulting from the
small round blue-cell tumour (SRBCT) microarray 
experiment~\citep{KHAN2001}.  The data set consists of
a training set of 64 vectors containing 2308 gene expressions.
The data contains four types of tumours denoted 
EWS, BL-NHL, NB, and RMS.
As performed in the two previous papers, the genes are ranked
by their respective amount of discriminative information 
according to their $F$-statistic
$$
  F = \frac{
    \frac{1}{k-1}\sum_{m=1}^k n_m(\bar{x}_m-\bar{x})^2
  }{
    \frac{1}{n-k}\sum_{m=1}^k (n_m-1)\hat{\sigma}_m^2
  }
$$
where $\bar{x}$ is the sample mean,
$k=4$ is the number of classes, $n=64$ is the sample size,
$n_m$ is the sample size of class $m$, and likewise, 
$\bar{x}_m$ and $\hat{\sigma}_m^2$ are, respectively,
the sample mean and variance of class $m$.
The top 40 and bottom 160 scoring genes were selected
to provide a mix of the most and least informative genes.

Table~\ref{tab:sbct} displays the results of applying
the four threshold estimators with cross validation, 
the PDS method, and our concentration-based thresholding
with the sub-Gaussian formula and 
with false positive rates of 10, 5, and 1 percent.  
The percentage of matrix entries that are retained for the 
most informative $40\times40$ block and the least informative
block are tabulated.  Depending on the chosen false positive rate,
our concentration-based estimators give similar results to
Soft and SCAD thresholding.
PDS is the least conservative of the methods as it keeps the
most entries.  Hard and Adaptive LASSO thresholding are the 
most aggressive methods.

\begin{table}
  \begin{center}
  \begin{tabular}{lcccc}
  \hline
non-zero (\%) & CoM 10\% & CoM 5\% & CoM 1\% & PDS \\
Informative   & 30.3\% & 25.6\% &  8.5\% & 47.3\% \\
Uninformative &  5.4\% &  2.7\% &  0.4\% & 15.6\% \\
  \hline
\phn & Hard & Soft & SCAD & Adpt \\
Informative   & 6.0\%&  24.7\%&  21.3\%&  9.9\%\\
Uninformative & 0.3\%&   2.3\%&   1.8\%&  0.7\%\\
  \hline
  \end{tabular}
  \end{center}
  \capt{
    \label{tab:sbct}
    The percentages of non-zero off-diagonal entries in the six 
    covariance estimates partitioned into two parts: the informative
    $40\times40$ block of the highest scoring genes;
    the uninformative remaining matrix entries.
  }
\end{table}

It is also worth noting that our method is computationally efficient enough
to consider the entire $2308\times2308$ matrix at once.
In fact, it took only 131.3 seconds to compute $\espar$
on an Intel i7-7567U CPU, 3.50GHz.  
In contrast, the PDS
method, which still has significantly faster run times than 
cross validating the threshold estimators, took over 101 minutes
to finish.
False positive rates of 5\%, 1\%, and 0.1\%
were tested.  The fraction of non-zero entries in 
$\espar$ was 8.6\%, 2.0\%, and 0.22\%, respectively. 
For comparison, the fraction of 
non-zero entries retained by PDS was 17.7\%.  
If such an analysis is meant to lead to follow-up research
on specific gene pairings, then culling as many false positives
as possible is of critical importance.
The sparse covariance estimator was 
partitioned into $12\times12$
blocks and the number of non-zero entries was tabulated for
each.  The results are displayed in Figure~\ref{fig:sbct}.

\begin{figure}
  \includegraphics[width=0.32\textwidth]{\PICDIR/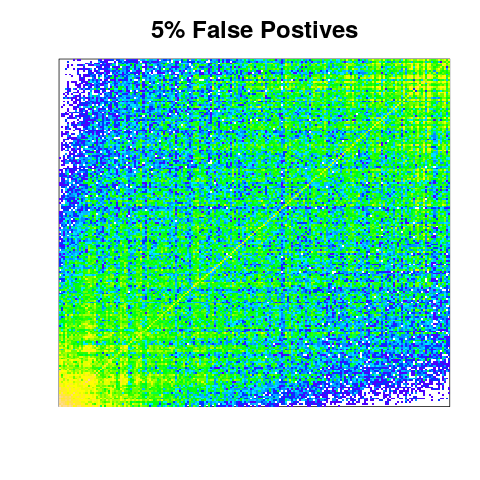}
  \includegraphics[width=0.32\textwidth]{\PICDIR/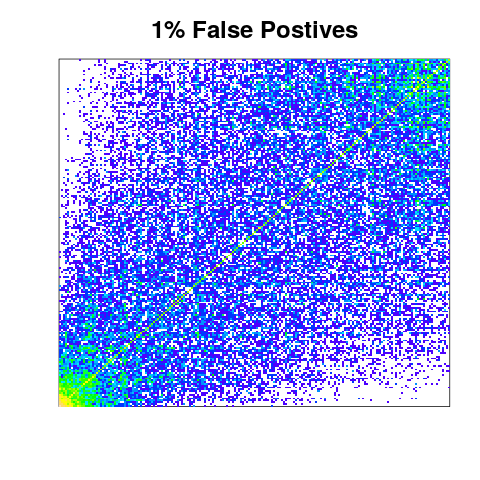}
  \includegraphics[width=0.32\textwidth]{\PICDIR/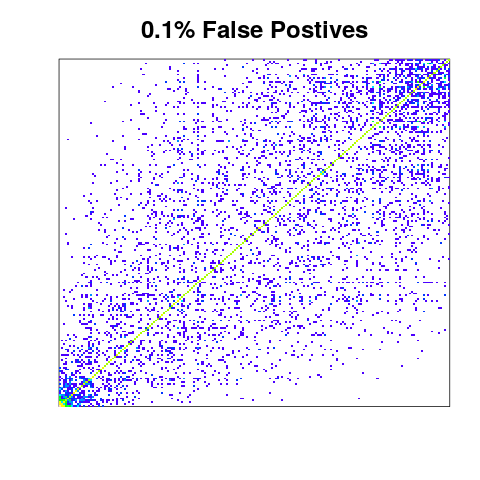}
  \capt{
    \label{fig:sbct} 
    A density plot of the number of non-zero 
    entries in $\espar\in\real^{2308\times2308}$ partitioned
    into $12\times12$ blocks for false positive rates of 
    5\%, 1\%, and 0.1\%.
  }
\end{figure}

\section{Supplementary Material}

The supplementary material consists of five sections.
The first parallels Section~\ref{sec:logconmeas} and 
considers sub-exponential measures and bounded random
variables as well as some additional simulations for 
multivariate Rademacher random variables. 
The second contains proofs of the lemmas and theorems
presented in the main article.  The third contains 
additional simulations motivating why our support 
recovery approach is better than past approaches.
The fourth contains derivations of Lipschitz 
coefficients for the functions used in Section~4.
The fifth is expository and contains past results 
from the concentraton of measure literature that
were directly used in this work.

\bibliographystyle{plainnat}
%\bibliography{\BIBDIR/kasharticle,\BIBDIR/kashbook,\BIBDIR/kashself,\BIBDIR/kashpack}
\bibliography{kasharticle,kashbook,kashself,kashpack}

\appendix
\section{Sub-Exponential and Bounded Data}

In line with our discussion of log concanve measures
in the main article, we include some information on
sub-exponential measures and data that is bounded.

\subsection{Sub-Exponential Distributions}
\label{sec:subexponential}

Compared with the previously discussed measures with sub-Gaussian 
concentration, there exists a larger class of measures with
sub-exponential concentration.  Such measures can be specified
as those that satisfy the Poincar{\'e} or spectral gap inequality
\citep{BOBKOV1997,LEDOUX2001,GOZLAN2010}.
  For a random variable $X$ on $\real^d$ with measure $\mu$, 
  this is
  $$
    \var{f(X)} \le C\int\abs{\nabla f}^2d\mu
  $$
  for some $C>0$ and for all locally Lipschitz functions $f$.

  If $X$ satisfies such an inequality, then---see 
  Theorem~S4.6%\ref{thm:spectralGap} 
  or Chapter 5 of \cite{LEDOUX2001}---for
  for $X_1,\ldots,X_n\in\real^d$ iid copies of $X$ and for
  some Lipschitz function $\phi:\real^{d\times n}\rightarrow\real$,
  $$
    \prob{
      \phi(X_1,\ldots,X_n) \ge \xv\phi(X_1,\ldots,X_n) + r
    } \le
    \exp\left(
      -\frac{1}{K} \min\left\{\frac{r}{b},\frac{r^2}{a^2}\right\}
    \right)
  $$
  where $K>0$ in a constant depending only on $C$ and 
  $$
    a^2 \ge \sum_{i=1}^n \abs{\nabla_i \phi}^2,~~~~~
    b   \ge \max_{i=1,\ldots,n} \abs{\nabla_i \phi}.
  $$

As in the log concave setting discussed in the main paper, 
$\phi$ is chosen to be 
$$
  \phi(X_1,\ldots,X_n) = 
  \norm*{ 
    \frac{1}{n}\sum_{i=1}^n (X_i-\xv X)\TT{(X_i-\xv X)}
  }_p^{1/2},
$$
which is Lipschitz with constant $n^{-1/2}$.  This results in 
values of $a^2 = 1$ and $b = n^{-1/2}$ for the above coefficients.
Hence, the radius in this setting is computed to be
$r_\alpha = \max\{ -K\log\alpha/\sqrt{n}, \sqrt{-K\log\alpha} \}$.
While an optimal (or reasonable) value for $K$ may not be known,
it makes little difference given the proposed 
procedure for choosing $\alpha$ detailed in the main paper
% Section~\ref{sec:falsePos}
for a desired false positive rate.  This is because the
term $-K\log\alpha$ will be equivalently tuned to determine 
the optimal size of the constructed confidence set.

As $r_\alpha$ in this setting is bounded below 
by a constant $\sqrt{-K\log\alpha}$, we do not achieve the 
nice convergence results as in the log concave setting.
However, the dimension-free concentration still allows for 
good performance in simulation settings as was seen in 
Section~6.%\ref{sec:numeric}.

\subsection{Bounded Random Variables}
\label{sec:boundedrv}

In this section, we consider random variables that are bounded
in some norm.  Consider a Banach space $(B,\norm{\cdot})$ and
a collection of iid random variables $X_1,\ldots,X_n\in B$
such that $\norm{X_i}\le U$ for all $i=1,\ldots,n$.
Given only this assumption, the bounded differences inequality,
detailed in the supplementary material %Appendix~\ref{app:boundDiff} 
and in Section 3.3.4 of
\cite{GINENICKL2015}, can be applied in this specific setting.
It provides sub-Gaussian concentration for such random variables.

  Specifically,
  let $X_1,\ldots,X_n\in\real^d$ be iid with 
  $\norm{X_i}_{\ell^2}\le U$ for $i=1,\ldots,n$.  Then,
  for any $p\in[1,\infty]$, $\norm{X_i\TT{X_i}}_p\le U^2$,
  and
  $$
    \prob{
      \norm*{ \eemp - \strue }_p \ge
      \xv\norm*{ \eemp - \strue }_p +
      r
    } 
    \le \ee^{-2nr^2/U^2}.
  $$
  This follows immediately from Theorem~S4.8.%\ref{thm:boundDiff}

Hence, for any collection of real valued random vectors bounded 
in Euclidean norm, the bounded differences inequality can be
applied to the  empirical estimate for any of the $p$-Schatten norms.
The radius is $r_\alpha = U\sqrt{(1/2n)\log\alpha}$.
However, unlike in the previous setting, the bounds may not 
necessarily be dimension free.  
\begin{example}[Distributions on the Hypercube]
If the components 
$\abs{X_{i,j}}\le1$ such as for multivariate uniform or Rademacher 
random variables,
then $U={d}^{1/2}$.  Consequently, 
$r_\alpha = O(\sqrt{d/n})$ is not dimension free.
While this makes estimation with respect to operator norm distance
challenging, we can still use Theorem~1 to %\ref{thm:falsepos} to 
fix the false positive rate.
\end{example}

\subsection{Simulations on High Dimensional Binary Vectors}

Random binary vectors fall into the category of bounded 
random variables, which have sub-Gaussian concentration as
a consequence of the bounded differences inequality---an
extension of \holders inequality---as 
discussed in Section~\ref{sec:boundedrv}.  
The result is a slightly different form for the confidence
balls compared with the log concave setting.
And while the concentration behaviour in this setting relies
on the dimension and is poor for producing an estimator
that is close in operator or Hilbert-Schmidt norm, 
our support recovery methodology is still able to perform
well in this setting.

Table~\ref{tab:fptpRad} displays false positive and 
true positive percentages for seven sparse estimators 
computed over 100 replications
of a random sample of size $n=50$ of $d=50,100,200,500$
dimensional 
multivariate Rademacher 
data with a tri-diagonal covariance matrix $\Sigma$ whose
diagonal entries are 1 and whose off-diagonal entries are 0.3.
As a consequence of the bounded differences inequality, this 
case also exhibits sub-Gaussian behaviour.  As a consequence,
Table~\ref{tab:fptpRad} is similar to Table~1 from the main article.
%\ref{tab:fptpGauss}.
Concentration estimators perform better as $d$ increases;
Threshold estimators are overly aggressive as $d$ increases;
And the PDS method's support recover is unaffected by the
change in $d$.

\begin{table}
  \begin{center}
  \begin{tabular}{l rrrr c rrrr}
  \hline
        & \multicolumn{4}{c}{False Positive \%} &  
        \multicolumn{4}{c}{True Positive \%}  \\
    Dimension & 50 & 100 & 200 & 500 && 50 & 100 & 200 & 500 \\
  \hline
    CoM 1\% &0.0&0.1&0.3&0.9&& 0.0& 7.2&17.5&30.7\\
    CoM 5\% &0.9&1.9&3.2&4.4&&28.9&41.1&49.0&54.5\\
    PDS     &3.4&3.4&3.4&3.4&&50.2&50.3&50.8&50.6\\
    Hard    &0.0&0.0&0.0&0.0&& 0.1& 0.0& 0.0& 0.0\\
    Soft    &2.4&0.8&0.2&0.0&&44.2&29.3&18.1& 9.5\\
    SCAD    &1.8&1.0&0.5&0.1&&40.3&33.3&24.1&13.1\\
    Adpt    &0.2&0.1&0.1&0.0&&16.5& 9.4& 5.5& 2.5\\
  \hline
  \end{tabular}
  \end{center}
  \capt{
    \label{tab:fptpRad}
    Percentage of false and true positives for multivariate Rademacher data
    and $\strue$ tri-diagonal with diagonal entries 1 and off-diagonal
    entries 0.3.
  }
\end{table}

\begin{figure}
  \begin{center}
    \includegraphics[width=0.6\textwidth]{\PICDIR/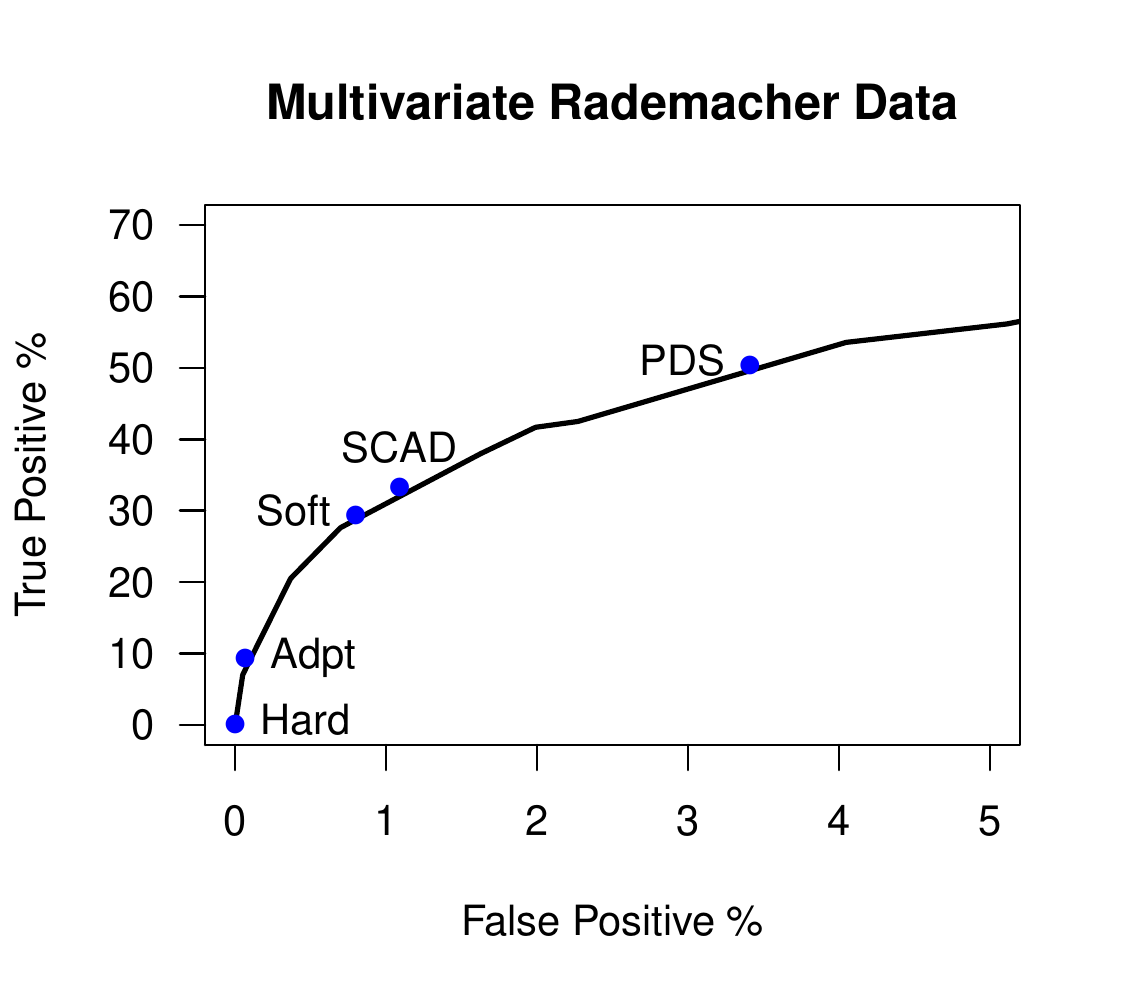}
  \end{center}
  \capt{
    \label{fig:fptpBoth}
    A line demarcating the trade-off between false and true positive
    recoveries for multivariate Gaussian (top left), Laplace (top right),
    and Rademacher (bottom)
    data from 100 replications of sample size $n=50$ and 
    dimension $d=100$.  The lines were formed by varying the false
    positive rate for the concentration estimator.  
    The blue dots indicate the performance of 
    past methods for sparse covariance
    estimation.
  }
\end{figure}

\section{Proofs}
\label{sec:proofs}

\begin{proof}[Proof of Lemma~1]%\ref{lem:quant}]
  We begin with the collection of $N = d(d-1)/2$ random variables
  $\hat{\sigma}_{i,j} = n^{-1}\sum_{k=1}^n X_{k,i}X_{k,j}$, which 
  we will denote $Z_1,\ldots,Z_N$.  Without loss of generality,
  assume that $Z_1,\ldots,Z_{N_0}$ have mean zero and 
  $Z_{N_0+1},\ldots,Z_{N_1+N_0}$ have nonzero mean and $N = N_0+N_1$.
  To achieve $\eta$ false positives, we would find the index $k_0$
  corresponding to the $\lfloor(1-\eta)N_0\rfloor$ order statistic
  of the $Z_1,\ldots,Z_{N_0}$, and set all entries 
  $\abs{Z_i}\le\abs{Z_{k_0}}$ to zero.  Instead, we find the
  index $\hat{k}$ corresponding to the 
  ${\lfloor (1-\eta) N \rfloor}$ order statistic of all the $Z_i$.

  Given that $\strue\in\mathcal{U}(\kappa,\delta)$, we have
  $$
    \abs{k_0-\hat{k}}\le \kappa d.
  $$
  Thus, when considering the achieved false positive rate 
  $\#\{\abs{Z_i} < \abs{Z_{\hat{k}}}|i\le N_0\}/N_0$ 
  to the target rate
  $\#\{\abs{Z_i} < \abs{Z_{k_0}}|i\le N_0\}/N_0$,
  we have
  $$
  \abs{\eta-\hat{\eta}}\le \frac{\kappa d}{N_0} =
  \frac{2\kappa}{d-1-2\kappa} = O(d^{\nu-1}).
  $$

\end{proof}

\begin{proof}[Proof of Lemma~2]% \ref{lem:covSymmet}]
  This proof follows from the result of
  \cite{LATALA2005} Theorem 2---also found in Theorem 2.3.8 
  of \cite{TAO2012}---without the assumption of iid entries
  in the random matrix but with many entries equal to zero.

  We first apply the expectation with respect to $\Veps$ and
  use the result from \cite{LATALA2005}.
  \begin{align*}
    \xv\norm{ A\circ\Veps }_\infty 
    &= 
    \xv_A\xv_\Veps\norm{ A\circ\Veps }_\infty \\
    &\le 
    \left[\xv_A
      \xv_\Veps\norm{ A\circ\Veps }_\infty^2
    \right]^{1/2} \\
    &\le 
    \left[{
      K_1\xv\max_{i=1,\ldots,d}\left(\sum_{j=1}^da_{i,j}^2\right) +
      K_2\xv\left(\sum_{i,j=1}^{d}a_{i,j}^4\right)^{1/2}
    }\right]^{1/2} 
  \end{align*}
  with $K_1,K_2$ universal constants.
  For the second term in the above equation, we have
  via Jensen's inequality and the fact that $\abs{a_{i,j}}\le1$
  that
  \begin{align*}
    \xv\left(\sum_{i,j=1}^{d}a_{i,j}^4\right)^{1/2}
    \le
    \left(\sum_{i,j=1}^{d}\xv a_{i,j}^4\right)^{1/2}
    \le
    (d^2\rho)^{1/2} = d\rho^{1/2}.
  \end{align*}
  For the first term in the above equation, we make use
  of the fact that $\abs{a_{i,j}}\le1$ and that only $\rho$ are
  non-zero resulting in%to replace the $a_{i,j}$ with 
  %$b_{i,j}\dist\distBern{\rho}$.
  \begin{align*}
    \xv\max_{i=1,\ldots,d}\left(\sum_{j=1}^da_{i,j}^2\right)
    %\le
    %\xv\max_{i=1,\ldots,d}\left(\sum_{j=1}^db_{i,j}^2\right)
    &\le 
    \xv\left(\sum_{i=1}^d(\sum_{j=1}^da_{i,j}^2)^2\right)^{1/2}
    \\ &\le
    \xv\left(
      \sum_{i,j}^da_{i,j}^2+\sum_{i,j\ne k=1}^da_{i,j}a_{i,k}
    \right)^{1/2}
    \\ &\le 
    \left(
      d^2\rho+(d^3-d^2)\rho^2
    \right)^{1/2}
    \\ &\le
      d\rho^{1/2}+d^{3/2}\rho.
  \end{align*}
  Combining the above results and updating the constants $K_1,K_2$
  as necessary gives the desired result
  $$
    \xv\norm{ A\circ\Veps }_\infty 
    \le
    \left[{
      K_1d\rho^{1/2} +
      K_2d^{3/2}\rho
    }\right]^{1/2} 
    \le
    K_1d^{1/2}\rho^{1/4} + K_2d^{3/4}\rho^{1/2}.
  $$
\end{proof}

\begin{proof}[Proof of Theorem~1]% \ref{thm:falsepos}]
  Without loss of generality, we can normalize $\eemp$ 
  such that the diagonal entries are 1.  Thus 
  $\eemp_0=I_d$, the $d$ dimensional identity matrix,
  and the off-diagonal entries of all matrices considered
  will be bounded in absolute value by one.

  For the empirical covariance estimator, 
  $\norm{\eemp-\eemp_0}_\infty=\norm{\eemp}_\infty-1$. 
  We can decompose $\eemp$ into three parts: the diagonal
  of ones; the off-diagonal terms corresponding to 
  $\sigma_{i,j}\ne0$; and the off-diagonal terms 
  corresponding to $\sigma_{i,j}=0$.  The number of 
  non-zero off-diagonal terms is bounded in each row/column 
  by $\kappa$.  Hence,
  $$
    \norm{\eemp-\eemp_0}_\infty 
    \le \norm{\eemp_{\ne0}}_\infty+\norm{\eemp_{=0}}_\infty
    \le \kappa+\norm{\eemp_{=0}}_\infty
  $$
  where $\eemp_{=0}$ has entries $\hat{\sigma}_{i,j}$
  such that $\xv\hat{\sigma}_{i,j}=0$.  

  Let the entrywise or Hadamard product of two similar
  matrices $A$ and $B$ be $A\circ B$ with entry $ij$th entry 
  $(a_{i,j}b_{i,j})_{i,j}$.
  For ease of notation, we denote $\Pi_0=\eemp_{=0}$.
  Let $\Pi_{1}$ be the result of randomly removing 
  half of the entries from $\Pi_0$, which is 
  $\Pi_{1} = \Pi_0\circ B$ where $B\in\{0,1\}^{d\times d}$ 
  is a symmetric random matrix with iid $\distBern{1/2}$ 
  entries.  Considering the corresponding symmetric Rademacher
  random matrix, $\Veps = 2B-1$, we then have that
  $$
    \xv\norm{\Pi_1}_\infty = 
    \xv\norm{\Pi_0\circ B}_\infty = 
    \frac{1}{2}\xv\norm{\Pi_0 \pm \Pi_0\circ \Veps}_\infty.
  $$
  where the $\pm$ comes from the symmetry of $\Veps$.
  Thus,
  $$
    \abs*{\xv\norm{\Pi_{1}}_\infty-\frac{1}{2}\xv\norm{\Pi_0}_\infty}
    \le
    \frac{1}{2}\xv\norm{\Pi_0\circ \Veps}_\infty.
  $$

  This idea can be iterated.  Let 
  $\Pi_{m} = \Pi_0\circ B_1\circ\ldots\circ B_m$ with
  the $B_i$ iid copies of $B$ from before.
  Then, similarly,
  \begin{align*}
    \xv\norm{\Pi_m}_\infty &\le
    \frac{1}{2}\xv\norm{\Pi_{m-1}}_\infty + 
    \frac{1}{2}\norm{\Pi_{m-1}\circ \Veps_m}_\infty \\
    \xv\norm{\Pi_m}_\infty &\ge
    \frac{1}{2}\xv\norm{\Pi_{m-1}}_\infty -
    \frac{1}{2}\norm{\Pi_{m-1}\circ \Veps_m}_\infty.
  \end{align*}
  Moreover,
  $$
    \abs*{\xv\norm{\Pi_m}_\infty-2^{-m}\xv\norm{\Pi_0}_\infty}
    \le
    \sum_{j=0}^{m-1} 2^{-m+j}\xv\norm{\Pi_j\circ \Veps}_\infty.
  $$
  Applying Lemma~3.5%\ref{lem:covSymmet} 
  $m$ times and updating
  universal constants $K_1,K_2$ as necessary results in 
  \begin{align*}
    \abs*{\xv\norm{\Pi_m}_\infty-2^{-m}\xv\norm{\Pi_0}_\infty}
    &\le
    \sum_{j=0}^{m-1} 2^{-m+j}\left[
      K_1d^{1/2}2^{-j/4} + K_2d^{3/4}2^{-j/2}
    \right]\\
    &\le
      K_1d^{1/2}2^{-m/4} + K_2d^{3/4}2^{-m/2}
  \end{align*}
  Thus, for $\rho=2^{-m}$, we have 
  $$
    \abs*{\xv\norm{\Pi_m}_\infty-\rho\xv\norm{\Pi_0}_\infty}
    \le
    K_1d^{1/2}\rho^{1/4} + K_2d^{3/4}\rho^{1/2}.
  $$

  We want to replace the $\Pi_m$ with $\eemp_\rho-\eemp_0$
  and similarly for $\Pi_0$.  The off-diagonal entries such that
  $\sigma_{i,j}\ne0$ can contribute at most $\kappa=o(d^{\nu})$, $\nu<1$,
  to the operator norm.  Hence,
  %We can reinsert the diagonal of 1's, using the fact that 
  %$\norm{\Pi+I_d}_\infty=\norm{\Pi}+1$, and get 
  \begin{multline*}
    \abs*{
      \xv\norm{\eemp_\rho-\eemp_0}_\infty-
      \rho\xv\norm{\eemp-\eemp_0}_\infty
    }
    \le
    K_1d^{1/2}\rho^{1/4} + K_2d^{3/4}\rho^{1/2} +
    (1+\rho)o(d^{\nu}).
  \end{multline*}
  We lastly apply the crude---but effective in the non-asymptotic 
  setting---bound $\norm{\eemp}_\infty\ge d/n$ almost surely.  Dividing
  by $\xv\norm{\eemp-\eemp_0}_\infty$ results in 
  $$
    \abs*{
      \frac{
        \xv\norm{\eemp_\rho-\eemp_0}_\infty
      }{
        \xv\norm{\eemp-\eemp_0}_\infty
      } - \rho}
    \le
    K_1nd^{-1/2}\rho^{1/4} + K_2nd^{-1/4}\rho^{1/2} +
    o(nd^{\nu-1}).
  $$
  Thus, we require $\nu<1/2$ to make the final term negligible 
  for large $d$ with respect to the others. 

  We can extend this result to arbitrary $\rho\in(0,0.5]$ by 
  using the simple observation that given such a $\rho$,
  there exists an $a\in\integer^+$ such that $2^a\rho\in[0.5,1)$.
  Therefore, setting $\eta=2^a\rho$ and replacing 
  $\eemp$ with the corresponding matrix $\eemp_\eta$ 
  from Lemma~3.1%\ref{lem:quant} 
  allows us to proceed as above.
\end{proof}

\begin{proof}[Proof of Theorem~2]% \ref{thm:logConConv}]
  From the derivation in Section~2%\ref{sec:sparseEstProc}, 
  we have that
  $$
    \prob{ 
      \norm*{\espar-\strue}_p^{1/2} \ge 
      \xv \norm*{\eemp-\strue}_p^{1/2} + 2r_\alpha
    } \le \alpha
  $$
  for any $\espar$ such that $\norm{\espar-\eemp}\le r_\alpha$.
  Writing $Z=\norm*{\eemp-\strue}_p$ and $Y=\norm*{\espar-\strue}_p$ 
  and squaring and rearranging the terms gives,
  \begin{align*}
    \text{P}( 
       Y \ge 
       \xv Z &+ 
      4r_\alpha (\xv Z)^{1/2} + 4r_\alpha^2
    )  \\
    &=\prob{ 
      Y \ge 
      \xv Z\left(  
        1 + 4r_\alpha (\xv Z)^{-1/2} + 
        4r_\alpha^2 (\xv Z)^{-1}
      \right)
    } \\
    &=\prob{ 
      Y \ge 
      \xv Z\left(  
        1 + 2r_\alpha (\xv Z)^{-1/2}  
      \right)^2
    } 
    \le \alpha
  \end{align*}
  Given the standard convergence result for the empirical
  covariance matrix that $\xv\norm{\eemp-\strue}_p = O(n^{-1/2})$
  and our definition of $r_\alpha = O(n^{-1/2}\sqrt{-\log\alpha})$, we 
  now have that
  \iffalse
  $$
    \prob{ 
      \norm*{\espar-\strue}_p^{1/2} \ge 
      O(n^{-1/4}) + 2O(\sqrt{-(\log\alpha)/n})
    } \le \alpha.
  $$
  \fi
  $$
    \prob{  
      \norm*{\espar-\strue}_p \ge
      O\left(n^{-1/2}( 1 + n^{-1/4}\sqrt{-\log\alpha} )^2\right)
    } \le \alpha,
  $$
  which holds for any $\espar$ 
  such that $\norm{\espar-\eemp}\le r_\alpha$.
\end{proof}
\begin{proof}[Proof of Theorem~3]% \ref{thm:logConSpar}]
  Let $\estar = \{\hat{\sigma}_{i,j}\Indc{\sigma_{i,j}\ne0}\}$ 
  be the result of a perfect thresholding of the
  empirical covariance estimator.  That is,
  $\estar$ has support identical to the true $\strue$ and non-zero
  entries that coincide with $\eemp$.  Furthermore, let 
  $\etild$ be some other overly-sparse covariance estimator 
  resulting from zeroing entries in $\eemp$, but with more zeros than
  $\strue$.  For a radius $r_\alpha$, $\espar$ is the sparsest
  element in the corresponding confidence ball.
  \begin{align} 
    &\prob{ \mathrm{supp}(\espar) \ne \mathrm{supp}(\strue) }
    \nonumber\\
    \nonumber
    &~~~~= \prob{ 
      \norm{\estar-\eemp}_\infty^{1/2}\ge r_\alpha \text{ or }
      \norm{\etild-\eemp}_\infty^{1/2}\le r_\alpha 
    }\\
    \label{eqn:splitSupport}
    &~~~~= \prob{ 
      \norm{\estar-\eemp}_\infty^{1/2}\ge r_\alpha
    }+\prob{
      \norm{\etild-\eemp}_\infty^{1/2}\le r_\alpha 
    },
  \end{align}
  which, assuming a large enough sample size $n$, 
  are the two mutually exclusive events that the estimator 
  with correct support $\estar$ is not in the ball of radius $r_\alpha$
  and that a sparser estimator $\etild$ is in the ball.

  For the first term in Equation~\ref{eqn:splitSupport}, we 
  show that the probability that a matrix with the correct support
  lying outside of the confidence set will tend to zero.
  \begin{align*}
    &\prob{ \norm*{\estar-\eemp}_\infty^{1/2}\ge r_\alpha }
    \\&~~~~\le \prob{
      \norm{\estar-\strue}_\infty^{1/2} +
      \norm{\eemp-\strue}_\infty^{1/2} \ge r_\alpha
    }\\
    &~~~~\le \prob{
      \norm{\estar-\strue}_\infty^{1/2} \ge r_\alpha/2
    } + \prob{ 
      \norm{\eemp-\strue}_\infty^{1/2} \ge r_\alpha/2
    } = (\mathrm{I}) + (\mathrm{II})
  \end{align*}

  For $(\mathrm{II})$, we have that 
  $\xv\norm{\eemp-\strue} = O(n^{-1/2})$
  and that $r_\alpha^2=O(n^{-1}\log\alpha)$.
  Let $Z = \norm{\eemp-\strue}_\infty^{1/2}$
  for simplicity of notation.  Then, using the concentration
  result already established for Lipschitz functions of log 
  concave measures,
  \begin{align*}
    (\mathrm{II}) 
    &= \prob{ Z \ge r_\alpha/2 } \\
    &= \prob{ Z \ge \xv Z + (r_\alpha/2 - \xv Z) } \\
    &\le \exp\left(
      -n( r_\alpha/2 - \xv Z )^2/2c_0
    \right) \le C\alpha^{1/4}
  \end{align*}
  for some positive $C = o(1)$.  

  For $(\mathrm{I})$,
  applying the Gershgorin circle theorem \citep{ISERLES2009} 
  to the operator norm
  gives
  \begin{align*}
    (\mathrm{I})
    &\le \prob{
      \left(
        \max_{i=1,\ldots,d} {\textstyle\sum_{j=1}^d}
        \abs{ \hat{\sigma}_{i,j} - \sigma_{i,j} }\Indc{ \sigma_{i,j}\ne0 }
      \right)^{1/2}
      \ge r_\alpha/2
    } \\
    &\le \prob{
      \max_{i,j=1,\ldots,d} \abs{\hat{\sigma}_{i,j}-\sigma_{i,j}}^{1/2}
      \abs{\mathrm{supp}_\text{col}(\strue)}^{1/2} \ge r_\alpha/2
    }
  \end{align*}
  where $\mathrm{supp_\text{col}}( \Sigma ) = 
  \max_{j=1,\ldots,d} \abs{\{(i,j):\sigma_{i,j}\ne0\}}$ is the 
  maximal number of non-zero entries in any given column. 
  From Proposition~\ref{prop:lipDp2}, we have that 
  $\norm{\eemp-\strue}_2^{1/2}$ is Lipschitz with constant $n^{1/2}$.
  As the squared Frobenius norm is equal to the sum of the squares
  of the entries of the matrix, we in turn have that the entries
  $\abs{\hat{\sigma}_{i,j}-\sigma_{i,j}}^{1/2}$ are also Lipschitz
  with constant $n^{1/2}$.  As the maximum of $d^2$ Lipschitz functions
  is also still Lipschitz, we get similarly to case $(\mathrm{II})$ that
  $(\mathrm{I})\le C\alpha^\veps$ for some $\veps>0$.
  
  For the second term in Equation~\ref{eqn:splitSupport}, we
  show that the probability of any sparser matrix than $\estar$ 
  existing in the confidence ball goes to zero.
  Let $\mathrm{supp}(\etild)\subset\mathrm{supp}(\strue)$.  Then, 
  there exists a pair of indices
  $(i_0,j_0)\in\mathrm{supp}(\strue)$ such that 
  $(i_0,j_0)\notin\mathrm{supp}(\etild)$.
  \begin{align*} 
    \prob{
      \norm{\etild-\eemp}_\infty^{1/2}\le r_\alpha 
    } 
    &\le \prob{
      \max_{i=1,\ldots,d}\sum_{j=1}^d\hat{\sigma}_{i,j}^2
      \Indc{\tilde{\sigma}_{i,j}=0}
      \le r_\alpha^4
    }\\
    &\le \prob{
      \max_{i=1,\ldots,d}\sum_{j=1}^d\hat{\sigma}_{i,j}^2
      \Indc{{\sigma}_{i,j}=0}
      + \hat{\sigma}_{i_0,j_0}^2
      \le r_\alpha^4
    }\\
    &\le \prob{
      \hat{\sigma}_{i_0,j_0}
      \le r_\alpha^2
    }
  \end{align*}
  We have that if $\sigma_{i,j}\ne0$ then $\abs{\sigma_{i,j}}>\delta>0$.
  Hence, 
  $
    \hat{\sigma}_{i_0,j_0} = 
    (\hat{\sigma}_{i_0,j_0}-{\sigma}_{i_0,j_0})+{\sigma}_{i_0,j_0} \ge
    o_p(n^{-1/2}) + \delta.
  $
  Meanwhile, $r_\alpha = O(n^{-1/2})$. Thus, 
  $
    \prob{
      \hat{\sigma}_{i_0,j_0}
      \le r_\alpha^2
    }\rightarrow 0
  $ as $n\rightarrow\infty$
  as long as $\delta = o(n^{-1})$.
\end{proof}

\section{Estimation with the Empirical Diagonal}
\label{app:empDiag}

In this section, we demonstrate that the 
distance in operator norm is an insufficient metric to 
use for the comparison of estimators for large sparse 
covariance matrices in the non-asymptotic setting.  
The operator norm's usage in past
research 
\citep{BICKELLEVINA2008,BICKELLEVINA2008A,KAROUI2008,
ROTHMAN2009} 
stems from the result that 
``convergence in operator norm implies convergence of the eigenvalues
and eigenvectors.'' 
However, this does not imply strong performance for finite samples.
We demonstrate this by showing that the naive empirical diagonal 
covariance matrix---that is,
  the estimator $\ediag$ with $\ediag_{i,j}=\eemp_{i,j}$ if $i=j$
  and $\ediag=0$
otherwise---performs better in operator norm for finite samples.
%and can even achieve the same sample convergence rate, $(\log d)/n$,
%as do the thresholding estimators.

The simulation study from \cite{ROTHMAN2009} was reproduced
where four threshold estimators---hard, soft, SCAD, and 
adaptive LASSO---were applied to estimating the covariance 
matrix for a sample of $n=100$ random normal vectors
in dimensions $d=30,100,200,500$ for three different models.
We consider models 1 and 2, which respectively are
autoregressive covariance matrices with entries
$\sigma_{i,j} = \rho^{\abs{i-j}}$ and moving average
covariance matrices with entries
$\sigma_{i,j} = \rho\indc{\abs{i-j}=1} + \indc{i=j}$.
In both cases, we set $\rho=0.3$.  
The simulations were replicated 100 times and averaged.
The results are displayed
in Table~\ref{tab:rothmanNorm} for multivariate Gaussian
data and in Table~\ref{tab:rothmanExp} for multivariate Laplace
data.

For multivariate Gaussian data, we see that SCAD thresholding
gives superior performance in operator norm distance 
until $d=200$ where it gives
comparable performance to the empirical diagonal matrix.
At $d=500$, the empirical diagonal now gives the best performance.
In the case of multivariate Laplace data, the empirical 
diagonal outperforms all of the thresholding methods
in all of the dimensions considered 
with respect to operator norm distance.
It is worth noting that theoretical results for these 
threshold estimators were only demonstrated for sub-Gaussian
data.

We understand that the performance of the threshold estimators 
improves asymptotically with increasing $n$ 
whereas the empirical diagonal will 
perform worse in the limit.  The main point to make is that for fixed
finite samples, as generally occur in practise, it is unwise to
claim an estimator's superiority based solely on the operator norm
distance.  Hence, we argue instead for support recovery of the 
true covariance matrix as the critical problem to solve in the
context of high dimensional sparse covariance estimation.

\begin{table}[t]
  \begin{center}
  \begin{tabular}{rrrrrrr}
    \hline
       & \multicolumn{6}{c}{\bf MA Matrix} \\
   $d$ & \multicolumn{1}{c}{Empirical} & \multicolumn{1}{c}{Diagonal} & \multicolumn{1}{c}{Hard} & \multicolumn{1}{c}{Soft} & \multicolumn{1}{c}{SCAD} & \multicolumn{1}{c}{LASSO} \\
    \hline
    30 & 1.32 (0.14)& 0.72 (0.04)& 0.72 (0.09)& 0.70 (0.06)& {\bf0.63} (0.07)& 0.64 (0.07) \\
   100 & 3.03 (0.19)& 0.77 (0.04)& 0.87 (0.10)& 0.85 (0.04)& {\bf0.73} (0.06)& 0.77 (0.06) \\
   200 & 4.92 (0.21)& {\bf0.79} (0.04)& 0.95 (0.11)& 0.91 (0.03)& {\bf0.79} (0.06)& 0.82 (0.05) \\
   500 & 9.73 (0.25)& {\bf0.83} (0.05)& 1.06 (0.11)& 0.98 (0.02)& 0.88 (0.06)& 0.88 (0.05) \\
    \hline\hline
       & \multicolumn{6}{c}{\bf AR Matrix} \\
   $d$ & \multicolumn{1}{c}{Empirical} & \multicolumn{1}{c}{Diagonal} & \multicolumn{1}{c}{Hard} & \multicolumn{1}{c}{Soft} & \multicolumn{1}{c}{SCAD} & \multicolumn{1}{c}{LASSO} \\
    \hline
    30 & 1.33 (0.16)& 0.90 (0.04)& 0.79 (0.10)& 0.83 (0.07)& {\bf0.74} (0.09)& 0.76 (0.09) \\
   100 & 3.06 (0.21)& 0.94 (0.03)& 0.95 (0.08)& 1.02 (0.04)& {\bf0.86} (0.05)& 0.92 (0.05) \\
   200 & 4.99 (0.21)& 0.95 (0.02)& 1.00 (0.09)& 1.09 (0.03)& {\bf0.92} (0.04)& 0.97 (0.03) \\
   500 & 9.80 (0.26)& {\bf0.97} (0.02)& 1.04 (0.08)& 1.16 (0.02)& 0.99 (0.04)& 1.04 (0.03) \\
    \hline
  \end{tabular}
  \end{center}
  \capt{
    \label{tab:rothmanNorm}
    Distances from six different covariance estimators to truth 
    in operator norm for sample size $n=100$, dimensions
    $d=30,100,200,500$, and observations drawn from a multivariate
    Gaussian distribution.  Standard deviations computed over the 100
    replications are in brackets.
  }
\end{table}

\begin{table}[t]
  \begin{center}
  \begin{tabular}{rrrrrrr}
    \hline
       & \multicolumn{6}{c}{\bf MA Matrix} \\
   $d$ & \multicolumn{1}{c}{Empirical} & \multicolumn{1}{c}{Diagonal} & \multicolumn{1}{c}{Hard} & \multicolumn{1}{c}{Soft} & \multicolumn{1}{c}{SCAD} & \multicolumn{1}{c}{LASSO} \\
    \hline
    30 & 2.23 (0.55)& {\bf0.86} (0.12)& 0.94 (0.23)& 0.93 (0.11)& 0.91 (0.22)& 0.90 (0.16) \\
   100 & 6.18 (1.50)& {\bf0.93} (0.13)& 1.17 (0.33)& 1.17 (0.25)& 1.31 (0.35)& 1.18 (0.29) \\
   200 & 11.41 (2.67)& {\bf0.98} (0.13)& 1.41 (0.39)& 1.33 (0.32)& 1.56 (0.43)& 1.36 (0.36)\\
   500 & 26.30 (5.68)& {\bf1.07} (0.17)& 2.01 (0.81)& 1.82 (0.61)& 2.24 (0.73)& 1.89 (0.64) \\
    \hline
    \hline
       & \multicolumn{6}{c}{\bf AR Matrix} \\
   $d$ & \multicolumn{1}{c}{Empirical} & \multicolumn{1}{c}{Diagonal} & \multicolumn{1}{c}{Hard} & \multicolumn{1}{c}{Soft} & \multicolumn{1}{c}{SCAD} & \multicolumn{1}{c}{LASSO} \\
    \hline
    30 & 2.34 (0.51)& {\bf0.96} (0.09)& 1.03 (0.20)& 1.05 (0.09)& 1.00 (0.20)& 1.02 (0.15) \\
   100 & 6.22 (1.48)& {\bf1.05} (0.09)& 1.27 (0.26)& 1.28 (0.17)& 1.34 (0.30)& 1.25 (0.19) \\
   200 & 11.44 (2.38)& {\bf1.05} (0.12)& 1.44 (0.40)& 1.42 (0.25)& 1.59 (0.38)& 1.40 (0.31) \\
   500 & 26.69 (5.10)& {\bf1.09} (0.10)& 1.95 (0.77)& 1.82 (0.65)& 2.19 (0.77)& 1.82 (0.65) \\
    \hline
  \end{tabular}
  \end{center}
  \capt{
    \label{tab:rothmanExp}
    Distances from six different covariance estimators to truth 
    in operator norm for sample size $n=100$, dimensions
    $d=30,100,200,500$, and observations drawn from a multivariate
    Laplace distribution.  Standard deviations computed over the 100
    replications are in brackets.
  }
\end{table}

\section{Derivations of Lipschitz constants}
\label{app:calcs}

The following lemmas and propositions establish that
specific functions used in the construction of confidence
sets are, in fact, Lipschitz functions.
%The first lemma demonstrates that positive semi-definite 
%matrices in the $p$-Schatten norms act, in some sense, 
%like non-negative real numbers.

\begin{lm}
  \label{lem:posdefTrace}
  Let $A$ and $B$ be two $d\times d$ real valued
  symmetric non-negative definite matrices.  Then,
  $$
    \norm{ A+B }_1 = \norm{A}_1 + \norm{B}_1
  $$
  where $\norm{\cdot}_1$ is the trace class norm.
\end{lm}
\begin{proof}
  By definition, $\norm{A}_1 = \tr{(A^*A)^{1/2}}$.  
  If $A$ is symmetric and non-negative definite, then
  $(A^*A)^{1/2}=A$.  Hence, if $A$ and $B$ are symmetric
  and positive definite, then so is $A+B$.  Therefore,
  $$
    \norm{A+B}_1 = \tr{A+B} = \tr{A} + \tr{B} = \norm{A}_1 + \norm{B}_1.
  $$
\end{proof}

\begin{prop}[Lipschitz for $p=1$]
  \label{prop:lipDp1}
  Assume that $X_1,\ldots,X_n\in\real^d$ and that $\xv X_i =0$ for
  $i=1,\ldots,n$.
  The function $\phi:\real^{d\times n}\rightarrow\real$ defined as
  $$
    \phi(X_1,\ldots,X_n) = \norm*{
      \frac{1}{n} \sum_{i=1}^n X\TT{X}
    }_1^{1/2}
  $$
  is Lipschitz
  with constant $n^{-1/2}$ with respect to the metric
  $
    d_{(2,2)}({\bf X},{\bf Y}) = 
    \left(\sum_{i=1}^n \norm*{ X_i - Y_i }_{\ell^2}^2\right)^{1/2}.
  $
\end{prop}
\begin{proof}
  Let $X_1,\ldots,X_n,Y_1,\ldots,Y_n\in\real^d$ with 
  $\xv X_i = \xv Y_i=0$ for all $i$ and denote 
  ${\bf X}=(X_1,\ldots,X_n)$ and ${\bf Y}=(Y_1,\ldots,Y_n)$.
  Making use of Lemma~\ref{lem:posdefTrace}, we have
  \begin{align*}
    &n(\phi({\bf X})-\phi({\bf Y}))^2 \\
    &~~~~= \norm*{ \sum_{i=1}^n X_i\TT{X_i} }_1 +
       \norm*{ \sum_{i=1}^n Y_i\TT{Y_i} }_1 -
      2\norm*{ \sum_{i=1}^n X_i\TT{X_i} }_1^{1/2} 
       \norm*{ \sum_{i=1}^n Y_i\TT{Y_i} }_1^{1/2} \\
    &~~~~= \sum_{i=1}^n\left(
         \norm*{ X_i }_{\ell^2}^2 + \norm*{ Y_i }_{\ell^2}^2
      \right) -
      2\left[ 
        \left(\sum_{i=1}^n \norm*{ X_i }_{\ell^2}^2\right)
        \left(\sum_{i=1}^n \norm*{ Y_i }_{\ell^2}^2\right)
      \right]^{1/2}\\
    &~~~~= \sum_{i=1}^n\left(
         \norm*{ X_i }_{\ell^2}^2 + \norm*{ Y_i }_{\ell^2}^2
      \right) -
      2\left[ 
        \sum_{i,j=1}^n \norm{X_i}_{\ell^2}^2\norm{Y_j}_{\ell^2}^2
      \right]^{1/2} \\
    &~~~~= \sum_{i=1}^n\left(
         \norm*{ X_i }_{\ell^2}^2 + \norm*{ Y_i }_{\ell^2}^2
      \right) \\
      &~~~~~~~~~~~~-
      2\left[ 
        \sum_{i< j} \left(
          \norm{X_i}_{\ell^2}^2\norm{Y_j}_{\ell^2}^2+
          \norm{X_j}_{\ell^2}^2\norm{Y_i}_{\ell^2}^2
        \right) +
        \sum_{i=1}^n\norm{X_i}_{\ell^2}^2\norm{Y_i}_{\ell^2}^2
      \right]^{1/2} \\
    &~~~~\le \sum_{i=1}^n\left(
         \norm*{ X_i }_{\ell^2}^2 + \norm*{ Y_i }_{\ell^2}^2
      \right) \\
      &~~~~~~~~~~~~-
      2\left[ 
        2\sum_{i< j} \left(
          \norm{X_i}_{\ell^2}\norm{Y_j}_{\ell^2}
          \norm{X_j}_{\ell^2}\norm{Y_i}_{\ell^2}
        \right) +
        \sum_{i=1}^n\norm{X_i}_{\ell^2}^2\norm{Y_i}_{\ell^2}^2
      \right]^{1/2} \\
    &~~~~\le \sum_{i=1}^n\left(
         \norm*{ X_i }_{\ell^2}^2 + \norm*{ Y_i }_{\ell^2}^2
      \right) -
      2 \sum_{i=1}^n \norm{X_i}_{\ell^2}\norm{Y_i}_{\ell^2}\\
    &~~~~\le \sum_{i=1}^n \left(
        \norm{X_i}_{\ell^2}-\norm{Y_i}_{\ell^2}
      \right)^2\\
    &~~~~\le \sum_{i=1}^n 
        \norm{X_i-Y_i}_{\ell^2}^2
  \end{align*}
\end{proof}

The next two lemmas are used to prove the Lipschitz constant
for the $p$-Schatten norms with $p=2$ and $p=\infty$,
respectively.  The first lemma is reminiscent of 
the Cauchy-Schwarz inequality
in the setting of the $2$-Schatten norm.
\begin{lm}
  \label{lem:schatten2}
  Let $X_1,\ldots,X_n,Y_1,\ldots,Y_n\in\real^d$.  Then,
  %for any $p,q\in[1,\infty]$ such that 
  %$p^{-1} + q^{-1} = 1$ with the usual convention that 
  %$q^{-1}=0$ if $q=\infty$, we have that
  for the Frobenius norm, 
  $$
    \norm*{\sum_{i=1}^nX_i\TT{Y_i}}_2 \le
    \norm*{\sum_{i=1}^nX_i\TT{X_i}}_2^{1/2}
    \norm*{\sum_{i=1}^nY_i\TT{Y_i}}_2^{1/2}.
  $$
\end{lm}
\begin{proof}
  For any matrix $M\in\real^{d\times d}$, 
  we have that $\norm{M}_2^2 = \tr{M\TT{M}}$.  Hence, starting from
  the left hand side of the desired inequality and applying the 
  Cauchy-Schwarz inequality gives us
  \begin{align*}
    \norm*{\sum_{i=1}^nX_i\TT{Y_i}}_2 
    &=  \tr{
      \sum_{i,j=1}^n X_i\TT{Y_i}Y_j\TT{X_j}
    }^{1/2} \\
    &=   \left(
      \sum_{i,j=1}^n \iprod{X_i}{X_j}\iprod{Y_i}{Y_j}
    \right)^{1/2}\\
    &\le \left(
      \left(\sum_{i,j=1}^n \iprod{X_i}{X_j}^2\right)^{1/2}
      \left(\sum_{i,j=1}^n \iprod{Y_i}{Y_j}^2\right)^{1/2}
    \right)^{1/2}\\
    &\le \left(
      \tr{\sum_{i,j=1}^n X_i\TT{X_i}X_j\TT{X_j}}^{1/2}
      \tr{\sum_{i,j=1}^n Y_i\TT{Y_i}Y_j\TT{Y_j}}^{1/2}
    \right)^{1/2}\\
    &\le \norm*{\sum_{i=1}^n X_i\TT{X_i}}_2^{1/2}
         \norm*{\sum_{i=1}^n Y_i\TT{Y_i}}_2^{1/2}\\
  \end{align*}
\end{proof}

\begin{lm}
  \label{lem:schattenI}
  Let $X_1,\ldots,X_n,Y_1,\ldots,Y_n\in\real^d$.  Then,
  for the operator norm, 
  $$
    \norm*{\sum_{i=1}^nX_i\TT{Y_i}}_\infty \le
    \norm*{\sum_{i=1}^nX_i\TT{X_i}}_\infty^{1/2}
    \norm*{\sum_{i=1}^nY_i\TT{Y_i}}_\infty^{1/2}.
  $$
\end{lm}
\begin{proof}
  Using the definition of the operator norm and the 
  Cauchy-Schwarz inequality, we have that
  \begin{align*}
    \norm*{\sum_{i=1}^nX_i\TT{Y_i}}_\infty 
    &= \sup_{v\in\real^d,~\norm{v}_{\ell^2}=1} 
       \sum_{i=1}^n \iprod{X_i}{v}\iprod{Y_i}{v} \\
    &\le \left(
      \sup_{v\in\real^d,~\norm{v}_{\ell^2}=1} 
      \sum_{i=1}^n\iprod{X_i}{v}^2
      \sup_{u\in\real^d,~\norm{u}_{\ell^2}=1}  
      \sum_{i=1}^n\iprod{Y_i}{u}^2
    \right)^{1/2}\\
    &=
    \norm*{\sum_{i=1}^nX_i\TT{X_i}}_\infty^{1/2}
    \norm*{\sum_{i=1}^nY_i\TT{Y_i}}_\infty^{1/2}.
  \end{align*}
\end{proof}

\begin{prop}[Lipschitz for $p=2$ or $p=\infty$]
  \label{prop:lipDp2}
  Assume that $X_1,\ldots,X_n\in\real^d$ and that $\xv X_i =0$ for
  $i=1,\ldots,n$.  Let $p\in[2,\infty]$.
  The function $\phi:\real^{d\times n}\rightarrow\real$ defined as
  $$
    \phi(X_1,\ldots,X_n) = \norm*{
      \frac{1}{n} \sum_{i=1}^n X_i\TT{X_i}
    }_p^{1/2}
  $$
  is Lipschitz
  with constant $n^{-1/2}$ with respect to the metric
  $
    d_{(2,2)}({\bf X},{\bf Y}) = 
    \left(\sum_{i=1}^n \norm*{ X_i - Y_i }_{\ell^2}^2\right)^{1/2}.
  $
\end{prop}
\begin{proof}
  To establish that $\phi$ is Lipschitz with the desired constant,
  we proceed by bounding the G{\^a}teaux derivative.  
  Let $p\in\{2,\infty\}$.%, let $q$ be such that $p^{-1}+q^{-1}=1$.
  For $h\in\real$ and any $X_1\ldots,X_n,Y_1,\ldots,Y_n\in\real^d$
  such that $\norm{\sum_{i=1}^n X_i\TT{X_i}}_p\ne0$ and
  $\norm{\sum_{i=1}^nY_i\TT{Y_i}}_p\ne0$,
  \begin{align*}
    &\sqrt{n} d\phi( X_1,\ldots,X_n ; Y_1,\ldots,Y_n )=\\
    &~~~~= \lim_{h\rightarrow0}\left(
         \frac{
           \norm*{\sum_{i=1}^n(X_i+hY_i)\TT{(X_i+hY_i)}}_p
           - \norm*{\sum_{i=1}^nX_i\TT{X_i}}_p 
         }{
           2\norm*{\sum_{i=1}^nX_i\TT{X_i}}_p^{1/2}
           \left(\sum_{i=1}^n\norm*{hY_i}_{\ell^2}^2\right)^{1/2}
         }
       \right)\\
    &~~~~\le \lim_{h\rightarrow0}\left(
         \frac{
           \norm*{\sum_{i=1}^n\left(
             hY_i\TT{X_i}+hX_i\TT{Y_i}+h^2Y_i\TT{Y_i}
           \right)}_p
         }{
           2\norm*{\sum_{i=1}^nX_i\TT{X_i}}_p^{1/2}
           \left(\sum_{i=1}^n\norm*{hY_i}_{\ell^2}^2\right)^{1/2}
         }
       \right)\\
    &~~~~\le \frac{
           \norm*{\sum_{i=1}^n\left(
             Y_i\TT{X_i}+X_i\TT{Y_i}
           \right)}_p
         }{
           2\norm*{\sum_{i=1}^nX_i\TT{X_i}}_p^{1/2}
           \left(\sum_{i=1}^n\norm*{Y_i}_{\ell^2}^2\right)^{1/2}
         }\\
    &~~~~\le \frac{
           \norm*{\sum_{i=1}^n X_i\TT{Y_i} }_p
         }{
           \norm*{\sum_{i=1}^nX_i\TT{X_i}}_p^{1/2}
           \norm*{\sum_{i=1}^nY_i\TT{Y_i}}_p^{1/2}
         }\\
  \end{align*}
  where we used the facts that, for $M\in\real^{d\times d}$, 
  $\norm{M}_p = \norm{\TT{M}}_p$, that 
  $$
    \sum_{i=1}^n\norm{Y_i}^2_{\ell^2} =
    \sum_{i=1}^n\norm{Y_i\TT{Y}_i}_p  \ge 
    \norm{\sum_{i=1}^n Y_i\TT{Y}_i}_p,
  $$
  and that
  $$
    \norm*{\sum_{i=1}^n\left(
      Y_i\TT{X_i}+X_i\TT{Y_i}
    \right)}_p 
    \le 2\norm*{\sum_{i=1}^n X_i\TT{Y_i} }_p.
    %\le 2\norm*{\sum_{i=1}^n X_i\TT{Y_i} }_2
  $$
  Applying Lemma~\ref{lem:schatten2} in the $p=2$ case and
  Lemma~\ref{lem:schattenI} in the $p=\infty$ case shows that 
  $\sqrt{n}d\phi(\cdot) \le 1$ for all $X_i$ with 
  $\norm*{\sum_{i=1}^nX_i\TT{X_i}}_2\ne0$.   With application of the
  Mean Value Theorem, we have the desired Lipschitz constant.

  In the case that $\norm*{\sum_{i=1}^nX_i\TT{X_i}}_p=0$, we also
  achieve the same Lipschitz constant.  Indeed, as $X_i\TT{X_i}$
  is positive semi-definite, the norm can only be zero if all 
  $X_i = \TT{(0,\ldots,0)}$.  Hence, for any $Y_1,\ldots,Y_n\in\real^d$,
  \begin{multline*}
    \sqrt{n}\abs{\phi(X_1,\ldots,X_n)-\phi(Y_1,\ldots,Y_n)}
    =\\= \norm*{\sum_{i=1}^n Y_i\TT{Y_i}}_p^{1/2}
    \le \left(\sum_{i=1}^n \norm{Y_i}_{\ell^2}^2\right)^{1/2}
    = \left(\sum_{i=1}^n \norm{X_i-Y_i}_{\ell^2}^2\right)^{1/2}.
  \end{multline*}
\end{proof}

It is conjectured that the function $\phi(\cdot)$ is 1-Lipschitz
for all $p\in[1,\infty]$, which follows immediately if 
Lemmas~\ref{lem:schatten2} and~\ref{lem:schattenI} can be 
expanded to similar results for all $p\in[1,\infty]$.

\section{Concentration Results}

The following is a brief expository section detailing results 
used and the associated references for the various concentration
of measure tools used throughout this work.  More details on these
topics can be found in \cite{LEDOUX2001,BOUCHERON2013,GINENICKL2015}.

\subsection{Concentration results for log concave measures}
\label{app:logConcave}

Gaussian concentration for log concave measures is established
via the following theorems.  In short, Theorem~\ref{thm:logconcave}
states that log concave measures satisfy a 
logarithmic Sobolev inequality, which bounds the entropy
of the measure;
see Definition~\ref{def:entropy}.  
Logarithmic Sobolev inequalities were first introduced in
\cite{GROSS1975}, and this result is due to \cite{BAKRYEMERY1984}. 
Following that, Theorem~\ref{thm:gaussconc} 
links the logarithmic Sobolev inequality
with Gaussian concentration.  
Finally, Corollary~\ref{thm:prodmeasure} 
extends this Gaussian concentration to product measures 
whose individual components satisfy logarithmic Sobolev inequalities
in a dimension-free way due to the subadditivity of the entropy.  

\begin{defn}[Entropy]
  \label{def:entropy}
  For a probability measure $\mu$ on a measurable space 
  $(\Omega,\mathcal{F})$ and for any non-negative measurable
  function $f$ on $(\Omega,\mathcal{F})$, the entropy is
  $$
    \ent{\mu}{f} = \int f\log fd\mu -
    \left(\int f d\mu\right) \log\left(\int f d\mu\right).
  $$
\end{defn}

\begin{thm}[\cite{LEDOUX2001}, Theorem 5.2]
  \label{thm:logconcave}
  Let $\mu$ be strongly log-concave on $\real^d$ for some $c>0$.
  Then, $\mu$ satisfies the logarithmic Sobolev inequality.
  That is, for all smooth $f:\real^d\rightarrow\real$,
  $$
    \ent{\mu}{f^2} \le \frac{2}{c}\int \abs{\nabla f}^2 d\mu.
  $$
\end{thm}

\begin{thm}[\cite{LEDOUX2001}, Theorem 5.3]
  \label{thm:gaussconc}
  If $\mu$ is a probability measure on $\real^d$ such
  that 
  $
    \ent{\mu}{f^2} \le \frac{2}{c}\int \abs{\nabla f}^2 d\mu,
  $
  then $\mu$ has Gaussian concentration.  That is,
  Let $X\in\real^d$ be a random variable with law $\mu$.
  Then, for all $1$-Lipschitz functions 
  $\phi:\real^d\rightarrow\real$ and for all $r>0$, 
  $$
    \prob{\phi(X) \ge \xv\phi(X)+r}\le \ee^{-cr^2/2}.
  $$
\end{thm}

\begin{thm}[\cite{LEDOUX2001}, Corollary 5.7]
  \label{thm:prodmeasure}
  Let $X_1,\ldots,X_n\in\real^d$ be random variables with measures 
  $\mu_1,\ldots,\mu_n$, which are all strongly log-concave
  with coefficients $c_1,\ldots,c_n$.  Let 
  $\nu = \mu_1\otimes\ldots\otimes\mu_n$ be the product measure
  on $\real^{d\times n}$.  Then,
  $$
    \ent{\nu}{f^2} \le \frac{2}{\min_{i}c_i}
    \int \abs{\nabla f}^2 d\nu.
  $$
\end{thm}

Combining Theorems~\ref{thm:prodmeasure}  
and~\ref{thm:gaussconc} immediately gives the
following corollary.
\begin{cor}
  \label{thm:prodconc}
  Let $X_1,\ldots,X_n\in\real^d$ have measures 
  $\mu_1,\ldots,\mu_n$, which are all strongly log-concave
  with coefficients $c_1,\ldots,c_n$.  Let 
  $\nu = \mu_1\otimes\ldots\otimes\mu_n$ be the product measure
  on $\real^{d\times n}$.  Then,
  for any $1$-Lipschitz $\phi:(\real^d)^n\rightarrow\real$
  and for any $r>0$,
  $$ 
    \prob{\phi(X_1,\ldots,X_n)\ge \xv\phi(X_1,\ldots,X_n)+r}
    \le \ee^{-\min_{i}c_ir^2/2}.
  $$
\end{cor}

\subsection{Concentration results for sub-exponential measures}

If the log Sobolev inequality from above is replaced with the weaker 
spectral gap or Poincar{\'e} inequality, then we have the
sub-exponential measures.

\begin{thm}[\cite{LEDOUX2001}, Corollary 5.15]
  \label{thm:spectralGap}
  Let $X$, a random variable on $\real^d$ with measure $\mu$, 
  satisfy the Poincar{\'e} 
  inequality
  $$
    \var{f(X)} \le C\int\abs{\nabla f}^2d\mu
  $$
  for some $C>0$ and for all locally Lipschitz functions $f$.
  Then, for $X_1,\ldots,X_n\in\real^d$ iid copies of $X$ and for
  some Lipschitz function $\phi:\real^{d\times n}\rightarrow\real$,
  $$
    \prob{
      \phi(X_1,\ldots,X_n) \ge \xv\phi(X_1,\ldots,X_n) + r
    } \le
    \exp\left(
      -\frac{1}{K} \min\left\{\frac{r}{b},\frac{r^2}{a^2}\right\}
    \right)
  $$
  where $K>0$ in a constant depending only on $C$ and 
  $$
    a^2 \ge \sum_{i=1}^n \abs{\nabla_i \phi}^2,~~~~~
    b   \ge \max_{i=1,\ldots,n} \abs{\nabla_i \phi}.
  $$
\end{thm}

\subsection{Concentration results for bounded random variables}
\label{app:boundDiff}

The following results can be found in more depth in 
\cite{GINENICKL2015} Section 3.3.4 and specifically 
in Example 3.3.13 (a).  
Theorem~\ref{thm:boundDiff} below is effectively a more 
general version of Hoeffding's Inequality.
To establish it, we begin
with the definition of functions of bounded differences.
\begin{defn}[Functions of Bounded Differences]
  A function $f:\real^{d\times n}\rightarrow \real$ is 
  of bounded differences if
  $$
    \sup_{x_i,x_i',x_j\in\real^d, j\ne i}
    \abs*{
      f(x_1,\ldots,x_n) - f(x_1,\ldots,x_i',\ldots,x_n)
    } \le c_i
  $$
\end{defn}
Then, Gaussian concentration can be established for functions 
of bounded differences by the following theorem.
\begin{thm}
  \label{thm:boundDiff}
  Let $X_1,\ldots,X_n\in\real^d$ and 
  $Z=f(X_1,\ldots,X_n)$ where $f$ has bounded differences
  with $c = \sum_{i=1}^n c_i$.  Then, for all $r>0$,
  $$
    \prob{Z\ge \xv Z + r} \le \ee^{-2r^2/c^2}.
  $$
\end{thm}
%%%%%%%%%%%%%%%%%%%%%%%%%%%%%%%%%%%%%%%%%%%%%%%%%%%%%%%
     
      %%%%%%   %%   %  %%%
      %        %%%  %  %  %
      %%%      % %% %  %   %
      %        %  %%%  %  %
      %%%%%%   %   %%  %%%

%%%%%%%%%%%%%%%%%%%%%%%%%%%%%%%%%%%%%%%%%%%%%%%%%%%%%%5
\end{document}